\def\mode{0}

\if 0\mode

\maketitle
\else

\documentclass[runningheads,envcountsame]{llncs}
\usepackage[british]{babel}
\usepackage[utf8]{inputenc}
\usepackage{graphicx}
\usepackage{xspace}

\title{Binary expression of ancestors\\ in the Collatz graph\thanks{Research supported by European Research Council (ERC) under the European Union’s Horizon
2020 research and innovation programme (grant agreement No 772766, Active-DNA project), and Science Foundation
Ireland (SFI) under Grant number 18/ERCS/5746.}} 

\author{Tristan Stérin
}

\authorrunning{T. Stérin}
\institute{Hamilton Institute and 
Department of Computer Science \\
 Maynooth University \\
\email{tristan.sterin@mu.ie}\quad 
\url{https://dna.hamilton.ie/tsterin/}}

\usepackage{float}
\usepackage{subcaption} 
\usepackage{tikz}

\usepackage{graphicx}

\usepackage{amsmath,amsfonts,amssymb,amsthm}
\usepackage{microtype,xspace,wrapfig,multicol,lstautogobble} 
\usepackage{xcolor,soul}

\newcommand{\norm}[1]{||#1||}

\newcommand{\mud}[3]{#1 \equiv #2 \text{ mod } #3}

\newcommand{\N}{\mathbb{N}}
\newcommand{\Z}{\mathbb{Z}}
\newcommand{\ZnZ}[1]{\Z/#1\Z}
\newcommand{\Q}{\mathbb{Q}}

\newcommand{\Pa}{\mathcal{P}}
\newcommand{\F}{\mathcal{F}}
\newcommand{\E}{\mathcal{E}}

\newcommand{\B}{\mathcal{B}}
\newcommand{\Bs}{\B^{*}}
\newcommand{\Predk}[1]{\text{Pred}_{k}(#1)}
\newcommand{\Pred}[2]{\text{Pred}_{#1}(#2)}
\newcommand{\regk}[1]{\texttt{reg}_k(#1)}
\newcommand{\reg}[2]{\texttt{reg}_{#1}(#2)}
\newcommand{\join}{\texttt{join}}

\newcommand{\I}{\mathcal{I}}
\newcommand{\R}[2]{\mathcal{R}_{#1}(#2)}
\newcommand{\Nt}{\mathcal{N}}

\newcommand{\theTree}{$(\alpha_{0,-1})$-tree}

\newcommand{\composeP}[2]{#1\,\cdot #2}
\usepackage[hidelinks]{hyperref}

\usepackage{fancyvrb}
\usepackage{algorithm}
\usepackage[noend]{algpseudocode}
\usepackage{listings}

\usepackage{thmtools}
\usepackage{thm-restate}
\usepackage{environ}
\makeatletter
\NewEnviron{restatethis}[1]{%
  \protected@write\@auxout{}{%
    \string\@restatetheorem{#1}{\detokenize\expandafter{\BODY}}%
  }%
  \begin{theorem}\label{#1}\BODY\end{theorem}%
}
\newcommand{\@restatetheorem}[2]{%
  \expandafter\gdef\csname restatethis@#1\endcsname{#2}%
}
\newcommand{\restate}[1]{%
  \begingroup
  \renewcommand{\thetheorem}{\ref{#1}}%
  \begin{theorem}[$\E\Predk x$ is regular]\csname restatethis@#1\endcsname\end{theorem}%
  \endgroup
}
\makeatother
\newcommand{\figStruct}{3\xspace}
\newcommand{\figGene}{4\xspace}

\newcommand{\lemTTT}{44\xspace}

\newcommand{\appT}{A\xspace}
\newcommand{\appFea}{B\xspace}
\newcommand{\appCode}{C\xspace}
\newcommand{\appLong}{D\xspace}

\begin{document}

\maketitle

\fi

\begin{abstract}The Collatz graph is a directed graph with natural number nodes and where there is an edge from node
$x$ to node $T(x)=T_0(x)=x/2$ if $x$ is even, or to node $T(x)=T_1(x)=\frac{3x+1}{2}$ if $x$ is odd.
Studying the Collatz graph in binary reveals complex message passing behaviors
based on carry propagation which seem to capture the essential dynamics and complexity of the
Collatz process.
We study the set $\E \text{Pred}_k(x)$ that contains the binary expression of any ancestor $y$ that reaches $x$ with a limited budget of $k$ applications of $T_1$. The set $\E \text{Pred}_k(x)$ is known to be regular, Shallit and Wilson [EATCS 1992]\nocite{Shallit1992The}. In this paper, we find that the structure of the Collatz graph naturally leads to the construction of a regular expression, $\reg{k}{x}$, which defines $\E \text{Pred}_k(x)$. Our construction, is exponential in $k$ which improves upon the doubly exponentially construction of Shallit and Wilson. Furthermore, our result generalises Colussi's work on the $x = 1$ case [TCS 2011] to any natural number $x$, and gives mathematical and algorithmic\footnote{Code available here: \url{https://github.com/tcosmo/coreli}}  tools for further exploration of the Collatz graph in binary.
\end{abstract}



\section{Introduction}\label{sec:intro}

\begin{figure}[t!]
\center
\begin{subfigure}[t]{0.4\textwidth}
\resizebox{0.7\hsize}{!}{
$\begin{array}[t]{r}
    \overbrace{101110011}^x \\
+ \ 101110011\boldsymbol{1} \\ \hline
    \underbrace{1000101101}_{T_1(x)}0
\end{array}$}
\caption{\small The sum $3x+ \boldsymbol{1}$ in binary. 
The number $x$ gets added to $2x+ {\bf 1}$ which in binary is the
left shift of $x$ to which ${\bf 1}$ is added.}\label{fig:col_bin1}
\end{subfigure}
\hfill
\begin{subfigure}[t]{0.54\textwidth}
\resizebox{0.7\hsize}{!}{
$\begin{array}[t]{r}
                0\bar{0}\bar{1}\bar{0}\bar{1}\bar{1}10\bar{0}\bar{1}\bar{1}\bar{0} \\ \hline
    10001011010\phantom{0} 
\end{array}$} \
\caption{\small The sum $3x+\boldsymbol{1}$ interpreted as: ``each bit of $x$ sums with its right neighbour and the neighbour's potential carry''. Carries are represented by dots. The $+\boldsymbol{1}$ part of the operation is embedded in a carry on a fictional $0$ to the right of the rightmost $1$ bit. The first step, at the rightmost end, reads: $1 + \bar{0}$ which produces an ouput of $0$ and transports the carry from to $\bar{1}$.}\label{fig:col_bin2}
\end{subfigure}
\caption{\small Two ways to interpret the operation $3x + 1$ in binary, illustrated on the number $x$ with binary representation $101110011$. The method shown in (b) highlights carry propagation.}\label{fig:col_bin}
\end{figure}

Let $\N=\{0,1,\dots\}$. The Collatz map, $T:\N\to\N$, is defined by $T(x) = T_{0}(x) = x/2$ if $x$ is even or $T(x) = T_{1}(x) = (3x+1)/2$ if $x$ is odd. The Collatz graph is the directed graph generated by $T$, nodes are all $x\in\N$ and arcs are $(x,T(x))$. This map, and its graph, have been widely studied (see surveys \cite{survey1} and \cite{survey2}) and research has been driven by a problem, open at least since the 60s: \textbf{the Collatz conjecture}. The conjecture states that, in the Collatz graph, any strictly positive natural number is a predecessor of $1$. In other words, any $x>0$ reaches $1$ after a finite number of $T$-iterations. As of 2020, the Collatz conjecture has been tested for all natural numbers below $2^{68}$ without any counterexample found~\cite{Barina2020}.

There has been a fruitful trend of studying the Collatz process in binary \cite{properties,nlin0502061,Shallit1992The,DBLP:journals/tcs/Colussi11,DBLP:journals/tcs/Hew16,complexity13,capco:hal-02062503}. That is because, the maps $T_0$ and $T_1$ have natural binary interpretations. The action of $T_0$ corresponds to shifting the input's binary representation to the right -- deleting a trailing 0. While writing $T_1(x) = (3x + 1)/2 = (x + (2x + 1))/2$ reveals an interesting mechanism. In binary, the operation $x+(2x+1)$ corresponds to adding $x$ to its left-shifted version where the least significant bit has been set to $1$, Figure~\ref{fig:col_bin1}. Equivalently $x+(2x+1)$ corresponds to each bit of $x$ being added to its right neighbour and the  potential carry being placed on that neighbour. The $+1$ part of the operation can be represented as a carry appearing \textit{ex nihilo} after the rightmost  $1$, Figure~\ref{fig:col_bin2}. The described mechanism results in the propagation of a carry within the binary representation of $x$: two consecutive $1$s create a carry while two consecutive $0$s absorb an incoming carry. Representing trajectories in the Collatz graph with the carry-annotated base 2 representation of Figure~\ref{fig:col_bin2} leads to complex, ``Quasi Cellular Automaton'' evolution diagrams which seem to encompass the overall complexity of the Collatz process \cite{cloney19873x+}. These cary-annoted evolution diagrams are studied in depth in \cite{Collatz2RP}.

Here, we ask the following question: for a given bit string $\omega$, what is the shape of the bit strings which ``degrade'' into $\omega$ under the action of the Collatz process? Said otherwise, for an arbitrary $x$, can we characterize the binary expansion of all $y$ which reach $x$ in the Collatz process? To answer that question, we find that it is natural to put a budget on the number of times the map $T_1$ is used (see Remark~\ref{rk:justi}) and we study $\E \text{Pred}_k(x)$ the set of binary expressions of all $y$ which reach $x$ by using the map $T_1$ exactly $k$ times and the map $T_0$ an arbitrary number of times. There is a high-level argument which shows that for each $x$ and $k$ the set $\E \text{Pred}_k(x)$ is regular \cite{Shallit1992The}: the binary interpretation of the $3x+1$ operation as shown in Figure~\ref{fig:col_bin2} can be performed by a 4-state, reversible, Finite State Transducer (which states correspond to symbols $0,\bar{0},1,\bar{1}$). In \cite{Shallit1992The}, the authors make the point that having a budget of $k$ on the map $T_1$ corresponds to iterating that transducer $k$ times. Since finite iterations of FSTs lead to regular languages, $\E \text{Pred}_k(x)$ is regular. However, while it gives regular structure to the set $\E \text{Pred}_k(x)$, from the point of view of regular expressions, their argument does not lead to a tractable representation of $\E \text{Pred}_k(x)$: they construct exponentially large FSTs in $k$, leading to doubly exponentially large regular expressions when using general purpose regular expression generation algorithm \cite{surveyRegexAuto}.

In this paper, we find that the knowledge about the binary structure of ancestors of $x$ is embedded in the \textit{geometry} of finite paths that reach $x$ in the Collatz graph. By geometry of a path, we mean the \textit{parity vector} \cite{10.2307/2322189,wirsching1998the,arithProg,Terras1976} associated to that path, which corresponds to looking at the path's elements modulo $2$ (Figure~\ref{fig:example_paths}). We find that there is a tight link between the shape of the parity vector and the binary expression of the first element on the path, which is an ancestor of $x$. Hence, we focus on characterizing the shapes of parity vectors of paths ending in $x$ and then translate those shapes into binary expressions of ancestors. The budget of $k$ applications of the map $T_1$ will translate to the constraint of having $k$ $1$s in the parity vectors we consider. Our main result exploits the mapping between constrained parity vectors and binary representation of ancestors at ``$T_1$-distance'' $k$, in order to construct a regular expression $\regk{x}$ which defines $\E \text{Pred}_k(x)$. As the number of possible shapes of constrained parity vectors grows exponentially with $k$, these regular expressions are big\footnote{\if 0\mode Appendix~\ref{app:long} \else Appendix~\appLong in \cite{Collatz1arxiv} \fi shows $\reg{4}{1}$ which gives an idea of how large the regular expressions get.} but only exponential in $k$ (against doubly exponential in $k$ in previous constructions \cite{Shallit1992The}):

\begin{restatable}{theorem}{thmain}
\label{th:main}

For all $x\in\N$, for all $k\in\N$ there exists a regular expression $\text{\emph{\texttt{reg}}}_k(x)$ that defines $\E\Predk x$. The regular expression $\text{\emph{\texttt{reg}}}_k(x)$ is structured as a tree with $2^{k}3^{k(k-1)/2}$ branches, alphabetic width $O(2^{k}3^{k(k+1)/2})$ and star height equal to $1$.

\end{restatable}

Our result generalises \cite{DBLP:journals/tcs/Colussi11} which focused on the case $x=1$. We claim that the framework we introduce is more general than \cite{DBLP:journals/tcs/Colussi11,DBLP:journals/tcs/Hew16} and that, in potential future work, it could easily be applied to generalisations of the Collatz map such as the $T_q$ maps\footnote{Defined, for $q$ odd, by $T_q(x) = x/2$ if $x$ is even or $T_q(x) = qx + 1$ if $x$ is odd. These maps are as mysterious as the Collatz map.} \cite{10.2307/2153499}. Our result improves \cite{Shallit1992The} by an exponential factor which makes our construction more fit for pratical use. We have implemented the construction of Theorem~\ref{th:main} \if 0\mode (see Appendix~\ref{app:code}) \else (see Appendix~\appCode in the expanded version~\cite{Collatz1arxiv}) \fi and claim that it gives a new exploratory tool for studying the Collatz process in binary. Indeed, for any $x$ and $k$, we can sample $\regk{x}$ in order to analyse the different mechanisms by which the Collatz process transforms an input string into the binary representation of $x$ in $k$ odd steps. Also, from our result, one can also easily sample from $\regk{x}$ the \emph{smallest} ancestor of $x$ at $T_1$-distance $k$ which suggests a new approach for future work in trying to understand how the Collatz process \textit{optimally encodes} the structure of $x$ in ancestors at $T_1$-distance $k$.

In future work, we plan to use our algorithm as a tool to further understand the dynamics of the Collatz process in binary. In particular, we are very much concerned by the question: ``Can the Collatz process compute?''. Indeed, direct generalisations of the Collatz process are known to have full Turing power \cite{Conway,Koiran1999,DBLP:conf/tamc/KurtzS07}. While the Collatz conjecture, by characterizing the long term behavior of any trajectory, seems to imply that there are some limitations on the computational power of the Collatz process, nothing is known. We believe that further studying carry propagation diagrams in the binary Collatz process can lead to answers on the computational power of the Collatz process and that, the tools built in this article can support that research.

\nocite{DBLP:journals/tcs/Hew16,DBLP:journals/tcs/Colussi11}
\nocite{Terras1976,10.2307/2322189,wirsching1998the}
\nocite{10.2307/2153499}
\nocite{Conway,Koiran1999,DBLP:conf/tamc/KurtzS07}


\section{Parity vectors and occurrences of parity vectors}\label{sec:un}
Let $\N=\{0,1,\dots\}$. We recall that the Collatz map $T:\mathbb{N}\to\mathbb{N}$, is defined by $T(x)=T_0(x)=x/2$ if $x$ is even or $T(x) = T_1(x) = (3x + 1)/2$ if $x$ is odd.
The concept of parity vector was introduced in \cite{Terras1976} (under the name \textit{encoding vector}) and used, for instance, in \cite{10.2307/2322189,wirsching1998the,arithProg}. While we work with the same concept, we introduce a slightly different representation\footnote{This is done both because, in this format, parity vectors can be represented nicely in the plane (see Figure~\ref{fig:example_paths}), and because binary strings will be omnipresent in Section~\ref{sec:first} and we don't want to confuse the reader with too many of them. } of parity vectors by using arrows $\downarrow$ and $\leftarrow$ instead of bits 0 and 1. In this Section, we introduce notation to manipulate \textit{occurrences} of parity vectors in the Collatz graph and reformulate a crucial result of \cite{wirsching1998the} in our framework (Theorem~\ref{th:struct_occ}).

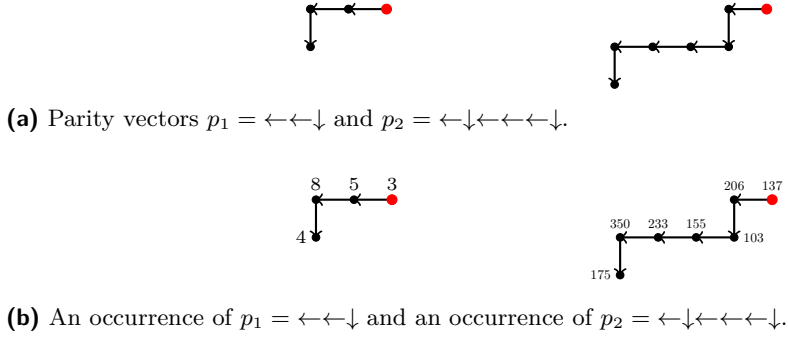
\begin{figure}[h!]



  \begin{subfigure}[t]{1.\linewidth}
  \centering
        \if 1\mode
        \begin{tikzpicture}[scale=0.4]
        \else
        \begin{tikzpicture}[scale=0.5]
        \fi

            \draw [thick] [->] (-0.1,0) -- (-1,0);
            \filldraw [red,fill=red] (0,0) circle [radius=0.13];
            \draw [fill] (-1,0) circle [radius=0.1];

            \draw [fill] (-1,0) circle [radius=0.1];
            \draw [thick] [->] (-1,0) -- (-2,0);
            \draw [fill] (-2,0) circle [radius=0.1];

            \draw [fill] (-2,0) circle [radius=0.1];
            \draw [thick] [->] (-2,0) -- (-2,-1);
            \draw [fill] (-2,-1) circle [radius=0.1];

            \begin{scope}[shift={(10,0)}]

                \draw [thick] [->] (-0.1,0) -- (-1, 0);
                \filldraw [red,fill=red] (0,0) circle [radius=0.13];
                \draw [fill] (-1, 0) circle [radius=0.1];

                \draw [thick] [->] (-1, 0) -- (-1, -1);
                \draw [fill] (-1, -1) circle [radius=0.1];

                \draw [thick] [->] (-1, -1) -- (-2, -1);
                \draw [fill] (-2, -1) circle [radius=0.1];

                \draw [thick] [->] (-2, -1) -- (-3, -1);
                \draw [fill] (-3, -1) circle [radius=0.1];

                \draw [thick] [->] (-3, -1) -- (-4, -1);
                \draw [fill] (-4, -1) circle [radius=0.1];

                \draw [thick] [->] (-4, -1) -- (-4, -2);
                \draw [fill] (-4, -2) circle [radius=0.1];

            \end{scope}
                  \end{tikzpicture}
    \caption{Parity vectors $p_1=\;\leftarrow \leftarrow \downarrow$ and $p_2 =\; \leftarrow \downarrow \leftarrow \leftarrow \leftarrow \downarrow$.}\label{fig:example_paths:a}
  \end{subfigure}
 \par\bigskip 
   \begin{subfigure}[t]{1\linewidth}
\centering
            \if 1\mode
        \begin{tikzpicture}[scale=0.4]
        \else
        \begin{tikzpicture}[scale=0.5]
        \fi

    \begin{scope}[shift={(0,0)}]

        \draw [thick] [->] (-0.1,0)-- (-1, 0);
                \filldraw [red,fill=red] (0,0) circle [radius=0.13];
        \draw [fill] (-1, 0) circle [radius=0.1];

        \draw [thick] [->] (-1, 0) -- (-2, 0);
        \draw [fill] (-2, 0) circle [radius=0.1];

        \draw [thick] [->] (-2, 0) -- (-2, -1);
        \draw [fill] (-2, -1) circle [radius=0.1];

        \node [above] at (0, 0) {${\scriptstyle 3 }$};
        \node [above] at (-1, 0) {${\scriptstyle 5 }$};
        \node [above] at (-2, 0) {${\scriptstyle 8 }$};
        \node [left] at (-2, -1) {${\scriptstyle 4 }$};

            \begin{scope}[shift={(10,0)}]

        \draw [thick] [->] (-0.1,0) -- (-1, 0);
                \filldraw [red,fill=red] (0,0) circle [radius=0.13];
        \draw [fill] (-1, 0) circle [radius=0.1];

        \draw [thick] [->] (-1, 0) -- (-1, -1);
        \draw [fill] (-1, -1) circle [radius=0.1];

        \draw [thick] [->] (-1, -1) -- (-2, -1);
        \draw [fill] (-2, -1) circle [radius=0.1];

        \draw [thick] [->] (-2, -1) -- (-3, -1);
        \draw [fill] (-3, -1) circle [radius=0.1];

        \draw [thick] [->] (-3, -1) -- (-4, -1);
        \draw [fill] (-4, -1) circle [radius=0.1];

        \draw [thick] [->] (-4, -1) -- (-4, -2);
        \draw [fill] (-4, -2) circle [radius=0.1];

        \node [above] at (0, 0) { \scalebox{0.5}{%
${ 137 }$}};
        \node [above] at (-1, 0) { \scalebox{0.5}{%
${ 206 }$}};
        \node [right] at (-1, -1) { \scalebox{0.5}{%
${ 103 }$}};
        \node [above] at (-2, -1) { \scalebox{0.5}{%
${ 155 }$}};
        \node [above] at (-3, -1) { \scalebox{0.5}{%
${ 233 }$}};
        \node [above] at (-4, -1) { \scalebox{0.5}{%
${ 350 }$}};
        \node [left] at (-4, -2) { \scalebox{0.5}{%
${ 175 }$}};
          \end{scope}
          \end{scope}
    \end{tikzpicture}

    \caption{\small An occurrence of $p_1=\;\leftarrow \leftarrow \downarrow$ and an occurrence of $p_2=\;\leftarrow \downarrow \leftarrow \leftarrow \leftarrow \downarrow$.}
    \label{fig:example_paths:b}
  \end{subfigure}
  \caption{\small Two parity vectors and one of their occurrences. In order to represent parity vectors, we use arrows $\downarrow$ and $\leftarrow$ instead of bits $0$ and $1$. When drawn in the plane, parity vectors read from \textbf{right} to \textbf{left}, start at the red dot.}\label{fig:example_paths}
\end{figure}

\begin{definition}[Parity Vector]\label{def:parvec} A parity vector $p$ is a word in $\{ \downarrow, \leftarrow \}^\ast$, i.e. a finite word, possibly empty, over the alphabet $\{ \downarrow, \leftarrow \}$. We call $\Pa$ the set of all parity vectors. The empty~parity~vector~is~$\epsilon$. We define~$\cdot$~to be the concatenation operation on parity vectors: $p = p_1 \cdot p_2$ is the parity vector consisting of the arrows of $p_1$ followed by the arrows of $p_2$. We use exponentiation in its usual meaning: $p^n = p \cdot p \ldots \cdot p$, $n$ times. 
\end{definition}
\begin{definition}[Norm and span]
As in \cite{wirsching1998the}, we define two useful metrics on parity vectors: (a) the \textit{norm} of $p$, written $\norm{p}$, is the total number of arrows in $p$ and (b) the \textit{span}\footnote{Called \textit{length} in \cite{wirsching1998the}. We change terminology to avoid confusion with the notion of length of a word over an alphabet. However, we keep the same mathematical notation $l(p)$.} of $p$, written $l(p)$, is the number of arrows of type $\leftarrow$ in $p$.
\end{definition}
\begin{definition}[Occurrence of a parity vector]\label{def:occ}
Let $p=a_0 \cdot \ldots \cdot a_{n-1}\in\Pa$ be a parity vector with $a_i \in \{\downarrow,\leftarrow\}$ and $n = \norm{p}$. An occurrence of $p$ in the Collatz graph, or, for short, an occurrence of $p$, is a $(n+1)$-tuple, $(o_0,\dots,o_{\norm{p}})\in\N^{\norm{p}+1}$ such that, for $0 \leq i < \norm{p}$, $o_{i+1} = T_{0}(o_i)$ if $a_i =\; \downarrow$ or $o_{i+1} = T_{1}(o_i)$ if $a_i = \; \leftarrow$. 
\end{definition}
\begin{definition}[Set of occurrences of a parity vector: $\alpha(p)$]\label{def:setocc}
Let $p\in\Pa$. We call $\alpha(p)$ the set of all the occurrences of the parity vector $p$. We order this set by the first number of each occurrence. Then, $\alpha_i(p)\in\N^{\norm{p}+1}$ denotes the $i^\text{th}$ occurrence of $p$ within that order and $\alpha_{i,j}(p)$, with $0 \leq j \leq \norm{p}$, denotes the $j^{\text{th}}$ term of the $i^{\text{th}}$ occurrence. In order to facilitate reading, we will write $\alpha_{i,-1}(p)$ instead of $\alpha_{i,\norm{p}}(p)$ to refer to the last element of the occurrence $\alpha_{i}(p)$. If the context clearly states the parity vector $p$ we will abuse notation and write $\alpha_{i,j}$ instead of $\alpha_{i,j}(p)$. 
\end{definition}
\begin{example}
Figure~\ref{fig:example_paths:a} shows two parity vectors in $\Pa$: $p_1=\;\leftarrow \leftarrow \downarrow$ and $p_2 =\; \leftarrow \downarrow \leftarrow \leftarrow \leftarrow \downarrow$. We have: $\norm{p_1} = 3$, $l(p_1) = 1$ and $\norm{p_2} = 6$, $l(p_2) = 4$.
In Figure~\ref{fig:example_paths:b}, it can be proved that we have $\alpha_0(p_1) = (3,5,8,4)$ and $\alpha_2(p_2) = (137, 206, 103, 155, 233, 350, 175)$.
\end{example}

\begin{definition}[Feasibility]\label{def:feasible}
A parity vector $p\in\Pa$ is said to be \textit{feasible} if it has at least one occurrence, i.e. if $\alpha_{0}(p)$ is defined. 
\end{definition}
The question ``Are all parity vectors feasible?'' is answered positively in \cite{wirsching1998the} (Lemma 3.1). This result is key to our work and we reformulate it in terms of occurrences of parity vectors:

\if 0\mode
\begin{figure}[t!]
  \centering

        \begin{tikzpicture}[scale=0.4]

            \draw [fill] (0,0) circle [radius=0.1];
            \draw [thick] [->] (0, 0) -- (-1, 0);
            \draw [fill] (-1, 0) circle [radius=0.1];

            \draw [thick] [->] (-1, 0) -- (-2, 0);
            \draw [fill] (-2, 0) circle [radius=0.1];

            \draw [thick] [->] (-2, 0) -- (-2, -1);
            \draw [fill] (-2, -1) circle [radius=0.1];

            \draw [thick] [->] (-2, -1) -- (-3, -1);
            \draw [fill] (-3, -1) circle [radius=0.1];

            \draw [thick] [->] (-3, -1) -- (-3, -2);
            \draw [fill] (-3, -2) circle [radius=0.1];

            \node [below] at (0, 0) {${\scriptstyle 11 }$};
            \node [right] at (-3, -2) {${\scriptstyle 10 }$};

            \draw [dashed] [->] (-3, -2)  to [out=-100,in=-100,looseness=1] (-6, -1);
            \node at (-5,-3) {${\scriptstyle+3^3}$};

            \draw [dashed] [->] (0, 0)  to [out=100,in=100,looseness=1] (-3, 1);
            \node [above] at (-1,1) {${\scriptstyle+2^5}$};

            \begin{scope}[xshift=-3cm,yshift=1cm]

                \draw [fill] (0,0) circle [radius=0.1];
                \draw [thick] [->] (0, 0) -- (-1, 0);
                \draw [fill] (-1, 0) circle [radius=0.1];

                \draw [thick] [->] (-1, 0) -- (-2, 0);
                \draw [fill] (-2, 0) circle [radius=0.1];

                \draw [thick] [->] (-2, 0) -- (-2, -1);
                \draw [fill] (-2, -1) circle [radius=0.1];

                \draw [thick] [->] (-2, -1) -- (-3, -1);
                \draw [fill] (-3, -1) circle [radius=0.1];

                \draw [thick] [->] (-3, -1) -- (-3, -2);
                \draw [fill] (-3, -2) circle [radius=0.1];

                \node [below] at (0, 0) {${\scriptstyle 43 }$};
                \node [right] at (-3, -2) {${\scriptstyle 37 }$};
                \draw [dashed] [->] (-3, -2)  to [out=-100,in=-100,looseness=1] (-6, -1);
                \node at (-5,-3) {${\scriptstyle+3^3}$};

                \draw [dashed] [->] (0, 0)  to [out=100,in=100,looseness=1] (-3, 1);
                \node [above] at (-1,1) {${\scriptstyle+2^5}$};

                \begin{scope}[xshift=-3cm,yshift=1cm]

                    \draw [fill] (0,0) circle [radius=0.1];
                    \draw [thick] [->] (0, 0) -- (-1, 0);
                    \draw [fill] (-1, 0) circle [radius=0.1];

                    \draw [thick] [->] (-1, 0) -- (-2, 0);
                    \draw [fill] (-2, 0) circle [radius=0.1];

                    \draw [thick] [->] (-2, 0) -- (-2, -1);
                    \draw [fill] (-2, -1) circle [radius=0.1];

                    \draw [thick] [->] (-2, -1) -- (-3, -1);
                    \draw [fill] (-3, -1) circle [radius=0.1];

                    \draw [thick] [->] (-3, -1) -- (-3, -2);
                    \draw [fill] (-3, -2) circle [radius=0.1];

                    \node [below] at (0, 0) {${\scriptstyle 75 }$};
                    \node [right] at (-3, -2) {${\scriptstyle 64 }$};
                    \draw [dashed] [->] (-3, -2)  to [out=-100,in=-100,looseness=1] (-6, -1);
                    \node at (-5,-3) {${\scriptstyle+3^3}$};

                    \draw [dashed] [->] (0, 0)  to [out=100,in=100,looseness=1] (-3, 1);
                    \node [above] at (-1,1) {${\scriptstyle+2^5}$};

                    \begin{scope}[xshift=-3cm,yshift=1cm]
                        \draw [fill] (0,0) circle [radius=0.1];
                        \draw [thick] [->] (0, 0) -- (-1, 0);
                        \draw [fill] (-1, 0) circle [radius=0.1];

                        \draw [thick] [->] (-1, 0) -- (-2, 0);
                        \draw [fill] (-2, 0) circle [radius=0.1];

                        \draw [thick] [->] (-2, 0) -- (-2, -1);
                        \draw [fill] (-2, -1) circle [radius=0.1];

                        \draw [thick] [->] (-2, -1) -- (-3, -1);
                        \draw [fill] (-3, -1) circle [radius=0.1];

                        \draw [thick] [->] (-3, -1) -- (-3, -2);
                        \draw [fill] (-3, -2) circle [radius=0.1];

                        \node [below] at (0, 0) {${\scriptstyle 107 }$};
                        \node [right] at (-3, -2) {${\scriptstyle 91 }$};
                    \end{scope}

                \end{scope}

            \end{scope}

        \end{tikzpicture}

\caption{\small Illustration of Theorem~\ref{th:struct_occ}. Structure of the set of occurrences of the parity vector $p=\;\leftarrow \leftarrow \downarrow \leftarrow \downarrow$, we have $l(p)=3$ and $\norm{p}=5$. For this parity vector $p$, we have $\alpha_{0,0}=11$ and $\alpha_{0,-1}=10$. As we can see, $\alpha(p)$ has a simple arithmetical structure.}\label{fig:G(s)}

\end{figure}
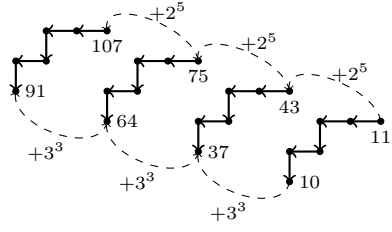
\fi 

\begin{restatable}[All parity vectors are feasible]{theorem}{thstructocc}
\label{th:struct_occ}
Let $p\in\Pa$. Then:
\begin{enumerate}
    \vspace{-0.5ex}
 \item $p \text{ is feasible i.e. } \alpha_{0} = (\alpha_{0,0},\ldots,\alpha_{0,-1})\in\N^{\norm{p}+1} \text{ is defined}$\label{point:so1}

\item $\alpha_{0,0} < 2^{\norm{p}}$ and $\alpha_{0,-1} < 3^{l(p)}$

\item \label{point:so3}  Finally we can completely characterize $\alpha_{i,0}$ and $\alpha_{i,-1}$ with: $\alpha_{i,0} = 2^{\norm{p}}i + \alpha_{0,0}$ and 
$\alpha_{i,-1} = 3^{l(p)}i + \alpha_{0,-1}$
\end{enumerate}
\vspace{-2ex}
\end{restatable}
\if 0\mode
\begin{proof}
This Theorem is essentially a reformulation of Lemma 3.1 in \cite{wirsching1998the}. We postpone the proof to Appendix~\ref{app:feasvec} because the concepts which the proof needs (introduced in \cite{wirsching1998the}) will not be used in the rest of this paper. 
\end{proof}
\else
\begin{proof}
This Theorem is essentially a reformulation of Lemma 3.1 in \cite{wirsching1998the}. We give the proof in \cite{Collatz1arxiv}, Appendix~\appFea. 
\end{proof}
\fi
\begin{example}
\if 0\mode Figure~\ref{fig:G(s)} \else Figure~\figStruct in \cite{Collatz1arxiv} \fi illustrates the knowledge that Theorem~\ref{th:struct_occ} gives on the structure of $\alpha(p)$, the set of occurrences\footnote{The result of \cite{arithProg} implies that one can prove the Collatz conjecture by only proving it for $\alpha_{i,j}(p)$ for all $i\in\N$, for any $p\in\Pa$, for any $0 \leq j \leq \norm{p}$.} of the parity vector $p$.
\end{example}
\section{First occurrence of parity vectors}\label{sec:first}
In this Section, we show that there is a direct link between the $\norm{p}$ arrows of a parity vector $p$ and the $\norm{p}$ bits of the binary representation  -- including potential leading $0$s -- of $\alpha_{0,0}(p)$ (Theorem~\ref{th:begin}). Then, we show that first occurrences of parity vectors can be arranged in a remarkably symmetric binary tree: the \theTree\ (Theorem~\ref{th:end}). As we work in binary, let's introduce some notation:
\begin{definition}[The set $\B^{\ast}$]
Let $\Bs$ be the set of finite (possibly empty) words written on the alphabet $\B=\{\texttt{0},\texttt{1}\}$. The empty word, is denoted by $\eta$. We define $\bullet$, the concatenation operator on these words and we use exponentiation in its usual meaning. Finally, for $\omega \in \Bs$, $|\omega|$ refers to the length (number of symbols) in the binary word $\omega$.
\end{definition}
\begin{definition}[The interpretations\footnote{We do not use the notation $[\![\cdot]\!]$ and its inverse $[\![\cdot]\!]^{-1}$ of \cite{DBLP:journals/tcs/Colussi11,DBLP:journals/tcs/Hew16} in order to avoid confusion. Indeed, in \cite{DBLP:journals/tcs/Colussi11,DBLP:journals/tcs/Hew16}, the use of this notation is meant to preserve leading $\texttt{0}s$ while we crucially need to control them in order to define the encoding function $\E$.} $\I$ and $\I^{-1}_{n}$]\label{def:inter}
Each word $\omega\in\Bs$ can, in a standard way, be interpreted as the binary representation of a number in $\N$. The function $\I:\Bs\to\N$ gives this interpretation. By convention, $\I(\eta) = 0$. Reciprocally, the partial function $\I^{-1}_{n}: \N \to \Bs_n$, where $\Bs_n$ is the set of $\omega\in\Bs$ with $|\omega| = n$, gives the binary representation of $x\in\N$ on $n$ bits. The value of $\I^{-1}_{n}(x)$ is defined only when $n \geq \lfloor \text{log}_2(2x+1) \rfloor$. We set $\I^{-1}_{0}(0) = \eta$. Finally, by $\I^{-1}(x)$ we refer to the binary representation of $x\in\N$ without any leading $\texttt{0}$. Formally, $\I^{-1}(x) = \I^{-1}_{\lfloor \text{log}_2(x) \rfloor + 1} (x) $ if $x\neq0$ and $\I^{-1}(0) = \I^{-1}_1(0) = \texttt{0}$.
\end{definition}
\begin{example}\label{ex:hiho}
$\I(\texttt{11}) = \I(\texttt{0011}) = 3$, $\I^{-1}(3) = \I^{-1}_{2}(3) = \texttt{11}$ and $\I^{-1}_{7}(3) = \texttt{0000011}$.
\end{example}
\subsection{Constructing $\alpha_{0,0}$}\label{sec:begin}
Let's notice that we have the following bijection (similarly introduced in \cite{Terras1976}):
\begin{restatable}{lemma}{lembij}
\label{lem:bij}
Define $\Pa_n = \{ p \in \Pa \text{ with } \norm{p} = n\}$. Then the function $f: \Pa_n \to \{0,\ldots,2^{n}-1\}$ defined by $f(p) = \alpha_{0,0}(p) $ is a bijection.
\end{restatable}
\if 0\mode
\begin{proof}
By cardinality, because $|\Pa_n| = |\{0,\ldots,2^n-1\}| = 2^n$, we just have to prove the injectivity of $f$. Let $p_1,p_2\in\Pa_n$ such that $f(p_1)=f(p_2)$. We write $p_1 = a_0 \cdot \ldots \cdot a_{n-1}$ and $p_2 = a'_0 \cdot \ldots \cdot a'_{n-1}$ with $a_i,a'_i\in\{\downarrow,\leftarrow\}$. Since $\alpha_{0,0}(p_1) = \alpha_{0,0}(p_2)$ and that the Collatz process is deterministic we deduce:
\begin{align*}
\alpha_{0,0}(p_1) &= \alpha_{0,0}(p_2)\\
\alpha_{0,1}(p_1) &= \alpha_{0,1}(p_2)\\
&\vdots\\
\alpha_{0,-1}(p_1) &= \alpha_{0,-1}(p_2)
\end{align*}
Thus, by Definition~\ref{def:occ} we deduce that $a_{i} = a'_{i}$ for $0 \leq i < n$. Thus $p_1=p_2$ which ends the proof.
\end{proof}
\else
\begin{proof}
\textbf{(Sketch)} By cardinality, only injectivity is to prove which comes by determinism of the Collatz process. Full proof in the expanded version \cite{Collatz1arxiv}.
\end{proof}
\fi

\if 0\mode
\input{figures/path_gene}
\fi

We can now define the \textit{Collatz encoding} of a parity vector $p\in\Pa$:
\begin{definition}[Collatz encoding of a parity vector $p$]\label{def:encode}
We define $\E: \Pa \to \Bs$ the \textit{Collatz encoding function} of parity vectors to be:
$ \E(p) = \I^{-1}_{\norm{p}}(\alpha_{0,0}(p))$.
The function $\E$ is well defined since, by Theorem~\ref{th:struct_occ}, $\alpha_{0,0}(p) < 2^{\norm{p}}$ . By Lemma~\ref{lem:bij}, $\E$ is bijective hence $\E^{-1}: \Bs \to \Pa$ is naturally defined.
\end{definition}
\begin{example}
$\E(p)$ is the binary representation of $\alpha_{0,0}(p)$ on $\norm{p}$ bits. We have: $\E(\downarrow\downarrow) = \texttt{00}$ or $\E(\leftarrow\downarrow\downarrow) = \texttt{101}$ \if 0\mode(see Figure~\ref{fig:order})\else (see Figure~\figGene in \cite{Collatz1arxiv})\fi. 
\end{example}
\begin{definition}[Admissibility of an arrow]
Let $a\in\{\downarrow,\leftarrow\}$. The arrow $a$ is said to be admissible for the number $x$ if and only if: ($a=\;\downarrow$ and $x$ is even) or ($a=\;\leftarrow$ and $x$ is odd).
\end{definition}
\begin{restatable}{lemma}{lemundeux}
\label{lem:undeux}
Let $p\in\Pa$ and $a\in\{\downarrow,\leftarrow\}$. Consider $\alpha_{0}(\composeP{p}{a}) = (\alpha_{0,0}(\composeP{p}{a}), \ldots, \alpha_{0,\norm{\composeP{p}{a}}}(\composeP{p}{a}))$. Then two cases:
\begin{itemize}
\item If $a$ is admissible for $\alpha_{0,-1}(p)$ then $(\alpha_{0,0}(\composeP{p}{a}), \ldots, \alpha_{0,\norm{p}}(\composeP{p}{a}))$ is the first occurrence of $p$, i.e. we have: $\alpha_{0}(p) = (\alpha_{0,0}(\composeP{p}{a}), \ldots, \alpha_{0,\norm{p}}(\composeP{p}{a}))$.
\item If $a$ is not admissible for $\alpha_{0,-1}(p)$ then $(\alpha_{0,0}(\composeP{p}{a}), \ldots, \alpha_{0,\norm{p}}(\composeP{p}{a}))$ is the second occurrence of $p$, i.e. we have: $\alpha_{1}(p) = (\alpha_{0,0}(\composeP{p}{a}), \ldots, \alpha_{0,\norm{p}}(\composeP{p}{a}))$.
\end{itemize}
\end{restatable}
\begin{proof}
\begin{itemize} 
\item If $a$ is admissible for $\alpha_{0,-1}(p)$ then $p \cdot a$ is forward feasible for $\alpha_{0,0}(p)$ and $(\alpha_{0,0}(p),\ldots,\alpha_{0,-1}(p), T(\alpha_{0,-1}(p)))$ is an occurrence of $p \cdot a$. It has to be the first occurrence of $p \cdot a$ otherwise, the existence of a lower occurrence of $p \cdot a$ would contradict the fact that $\alpha_{0}(p) = (\alpha_{0,0}(p),\ldots,\alpha_{0,-1}(p))$ is the first occurrence of $p$.
\item If $a$ is not admissible for $\alpha_{0,-1}(p)$, consider $\alpha_{0}(p\cdot a) = (o_0,\ldots,o_{\norm{p}+1})$ the first occurrence of $p\cdot a$. Then $(o_0,\ldots,o_{\norm{p}})$ is an occurrence of $p$. It cannot be the first one since the first occurrence of $p$ is followed by an arrow admissible for $\alpha_{0,-1}(p)$. However, by Theorem~\ref{th:struct_occ} we know that $o_0 = \alpha_{0,0}(p\cdot a) < 2^{\norm{p\cdot a}} = 2^{\norm{p} + 1} = 2 * 2^{\norm{p}}$. Thus we conclude that $(o_0,\ldots,o_{\norm{p}})$ is the second occurrence of $p$, i.e. $o_0 = \alpha_{1,0}(p) = 2^{\norm{p}} + \alpha_{0,0}(p)$ since for all $i\geq 2,\, \alpha_{i,0}(p) \geq 2*2^{\norm{p}}$ by Theorem~\ref{th:struct_occ}.
\end{itemize}
\end{proof}
\begin{restatable}[Recursive structure of $\E$]{theorem}{thbegin}
\label{th:begin}
Let $n\in\N$.
We have $\E(\epsilon) = \eta$. Then, for $p\in\Pa_{n}$ and $a \in \{\downarrow,\leftarrow\}$ we have
    $\E(p \cdot a ) = \textup{\texttt{0}}\bullet \E(p)$  if  $a$ is admissible for $\alpha_{0,-1}(p)$ and $\E(p \cdot a ) = \textup{\texttt{1}}\bullet \E(p)$ otherwise. 
\end{restatable}
\begin{proof}
By Definition~\ref{def:encode}, we have $\E(\epsilon) = \I^{-1}_{0}(\alpha_{0,0}(\epsilon)) = \I^{-1}_{0}(0)$ and $ \I^{-1}_{0}(0) = \eta$ by Definition~\ref{def:inter}. Hence, $\E(\epsilon) = \eta$. Now, let $p\in\Pa_{n}$, $a \in \{\downarrow,\leftarrow\}$. Two cases: 
\begin{itemize}
        \item If $a$ is admissible for $\alpha_{0,-1}(p)$, by Lemma~\ref{lem:undeux} we have $\alpha_{0,0}(p\cdot a) = \alpha_{0,0}(p)$. Thus we get that $\I^{-1}_{n+1}(\alpha_{0,0}(p \cdot a )) = \texttt{0}\bullet \I^{-1}_{n}(\alpha_{0,0}(p))$ since prepending a $\texttt{0}$ to a binary string doesn't change the number it represents. Hence, $\E(p \cdot a ) = \texttt{0}\bullet \E(p)$.

        \item If $a$ is not admissible for $\alpha_{0,-1}(p)$, by Lemma~\ref{lem:undeux} and Theorem~\ref{th:struct_occ} we get $\alpha_{0,0}(p\cdot a) = \alpha_{1,0}(p) = 2^{\norm{p}} + \alpha_{0,0}(p)$ which corresponds to prepending a bit $\texttt{1}$ to the binary representation of $\alpha_{0,0}(p)$ on $n$ bits. We conclude that $ \I^{-1}_{n+1}(\alpha_{0,0}(p \cdot a )) = \texttt{1}\bullet \I^{-1}_{n}(\alpha_{0,0}(p))$. Hence, $\E(p \cdot a ) = \texttt{1}\bullet \E(p)$.
\end{itemize}
\end{proof}
\begin{example}
\if 0\mode Figure~\ref{fig:order} \else Figure~\figGene in \cite{Collatz1arxiv} \fi
 illustrates Theorem~\ref{th:begin} on parity vectors of $\Pa_0,\Pa_1,\Pa_2,\Pa_3$.
\end{example}
\subsection{Constructing $\alpha_{0,-1}$}\label{sec:end}
 Theorem~\ref{th:begin} relies on knowing $\alpha_{0,-1}$ at each step in order to deduce the admissibility of the arrow which is being added. In this Section, we show that $\alpha_{0,-1}$ can also be recursively constructed. That construction will lead to a binary tree, the \theTree\, in which each node corresponds to the first occurrence of a parity vector. The symmetries of this tree will be crucial to our main result, Theorem~\ref{th:main}. The construction of $\alpha_{0,-1}$ relies on some elementary knowledge about groups of the form $\ZnZ{3^k}$ and their multiplicative subgroup $(\ZnZ{3^k})^*$. We recall the definition and main properties of these objects in \if 0\mode Appendix~\ref{app:structT}\else Appendix~\appT of \cite{Collatz1arxiv}\fi. In particular, we use the notation $2^{-1}_k$ to refer to the modular inverse of $2$ in $\ZnZ{3^k}$. Importantly, $2^{-1}_k$ is a primitive root of $(\ZnZ{3^k})^*$.  Those groups play an important role in our context because of the following result:
\begin{restatable}{lemma}{lemstar}
\label{lem:star}
Let $p\in\Pa$. Then $l(p) \neq 0 \Leftrightarrow \alpha_{0,-1}(p) \in (\ZnZ{3^{l(p)}})^{*}$. If $l(p)=0$, $\alpha_{0,-1}(p) = 0$.
\end{restatable}
\begin{proof}
We prove both directions:
$\Rightarrow$: we suppose $l(p) \neq 0$. We know $\alpha_{0,-1}(p) < 3^{l(p)}$ (Theorem~\ref{th:struct_occ}). We have to prove that $\alpha_{0,-1}(p)$ is not a multiple of three. The predecessor set of $y$, a multiple of $3$, in the Collatz graph is reduced to $\{2^n y \text{ for } n\in\N\}$. Indeed, we know that all $2^n y$ are predecessors of $y$ by the operator $T_0$. Furthermore, the operator $T^{-1}_1(y) = (2y-1)/3$ never yields to an integer if inputed a multiple of three and all $2^n y$ are. Hence no parity vector $p$ with $l(p)>0$ can have an occurrence ending in a multiple of three and we have the result. 
$\Leftarrow$: if $l(p)=0$ then $(\ZnZ{3^{l(p)}})^{*} = \emptyset$ so we have the result.
 If $l(p) = 0$ then $p$ has the form $p=\;(\downarrow)^{n}$. By Theorem~\ref{th:begin}, we deduce $\alpha_{0,0}(p) = 0$. Hence $\alpha_{0,-1}(p) = T^n(0) = 0$.
\end{proof}
We can recursively construct $\alpha_{0,-1}$ with the analogous of $T_0$ and $T_1$ in $\ZnZ{3^k}$:
\if 1\mode
\vspace{-3ex} \fi
\begin{restatable}[$T_{0,k}$ and $T_{1,k}$]{definition}{defT}\label{def:T}
The functions $T_{0,k}: \ZnZ{3^k} \to \ZnZ{3^k}$ and $T_{1,k}: \ZnZ{3^k} \to \ZnZ{3^k}$ are defined by:
$
T_{0,k}(x) = 2^{-1}_{k}x \text{ and } T_{1,k}(x) = 2^{-1}_{k}(3x+1)
$.
\end{restatable}
\begin{restatable}[Recursive structure of $\alpha_{0,-1}$]{theorem}{thend}
\label{th:end}
Let $n\in\N$. We have $\alpha_{0,-1}(\epsilon) = 0$. Then, for some $p\in\Pa_{n}$ and $k=l(p)$ we have $\alpha_{0,-1}(\composeP{p}{\downarrow}) = T_{0,k}(\alpha_{0,-1}(p))$ and $\alpha_{0,-1}(\composeP{p}{\leftarrow}) = T_{1,k+1}(\alpha_{0,-1}(p))$.
\end{restatable}
\if 1\mode
\begin{proof}
\textbf{(Sketch)} By induction, using that $T_{0,k}(x)=(3^k+x)/2$ when $x$ is odd and $T_{1,k+1}(x) = (3^k + 3x + 1)/2$ when $x$ is even. Full proof in \cite{Collatz1arxiv}.
\end{proof}
\else
\begin{proof}
Since any $x\in\N$ is an occurrence of the parity vector $\epsilon$, we have $\alpha_{0}(\epsilon) = (0,)$ (tuple with one element). Hence $\alpha_{0,-1}(\epsilon) = \alpha_{0,\norm{\epsilon}}(\epsilon) = \alpha_{0,0}(\epsilon) = 0$. Now, let $p\in\Pa_n$ for some $n\in\N$ and $k = l(p)$. Then notice that Equations~(\ref{eq:down}) and (\ref{eq:up}) are well defined because of Theorem~\ref{th:struct_occ}. Indeed, we know that $\alpha_{0,-1}(p) < 3^{l(p)} = 3^k$ thus $\alpha_{0,-1}(p)\in\ZnZ{3^k}$ and $\alpha_{0,-1}(p)\in\ZnZ{3^{k+1}}$ and we can use the operators $T_{0,k}$ and $T_{0,k+1}$ on it. 
\\
Let's consider $\alpha_{0,-1}(\composeP{p}{a})$ with $a\in\{\downarrow,\leftarrow\}$. Two cases:
\begin{itemize}
\item The arrow $a$ is admissible for $\alpha_{0,-1}(p)$: in that case, by Lemma~\ref{lem:undeux}, we know that $(\alpha_{0,0}(\composeP{p}{a}), \alpha_{0,1}(\composeP{p}{a}), \ldots, \alpha_{0,\norm{p}}(\composeP{p}{a}))$ is the first occurrence of $p$. Hence, $\alpha_{0,-1}(p) = \alpha_{0,\norm{p}}(\composeP{p}{a})$ and we have $\alpha_{0,-1}(\composeP{p}{a}) = T(\alpha_{0,\norm{p}}(\composeP{p}{a})) = T(\alpha_{0,-1}(p)) = T_{i}(\alpha_{0,-1}(p))$. With $i=0$ if $a=\;\downarrow$ or $i=1$ if $a=\;\leftarrow$. Then two cases:
\begin{enumerate} 
  \item If $a=\;\downarrow$ then $\alpha_{0,-1}(p)$ is even and $\alpha_{0,-1}(\composeP{p}{a}) = T_0(\alpha_{0,-1}(p)) = T_{0,k}(\alpha_{0,-1}(p))$  by Appendix~\ref{app:structT}, Lemma~\ref{lem:t0}.
  \item If $a=\;\leftarrow$ then $\alpha_{0,-1}(p)$ is odd and $\alpha_{0,-1}(\composeP{p}{a}) = T_1(\alpha_{0,-1}(p)) = T_{1,k+1}(\alpha_{0,-1}(p))$  by Appendix~\ref{app:structT}, Lemma~\ref{lem:t1}.
\end{enumerate}
 \item The arrow $a$ is not admissible for $\alpha_{0,-1}(p)$: in that case, by Lemma~\ref{lem:undeux}, we know that $(\alpha_{0,0}(\composeP{p}{a}), \alpha_{0,1}(\composeP{p}{a}), \ldots, \alpha_{0,\norm{p}}(\composeP{p}{a}))$ is the second occurrence of $p$. Hence, by Theorem~\ref{th:struct_occ}, $\alpha_{0,\norm{p}}(\composeP{p}{a}) = 3^k + \alpha_{0,-1}(p)$. Now, $\alpha_{0,-1}(\composeP{p}{a}) = T(\alpha_{0,\norm{p}}(\composeP{p}{a})) = T(3^k + \alpha_{0,-1}(p))$. Then two cases:
\begin{enumerate} 
  \item If $a=\;\downarrow$ then $\alpha_{0,-1}(p)$ is odd and $\alpha_{0,-1}(\composeP{p}{a}) = T_0(3^k + \alpha_{0,-1}(p)) = \frac{3^k + \alpha_{0,-1}(p)}{2} = T_{0,k}(\alpha_{0,-1}(p))$ by Appendix~\ref{app:structT}, Lemma~\ref{lem:t0}.
  \item If $a=\;\leftarrow$ then $\alpha_{0,-1}(p)$ is even and $\alpha_{0,-1}(\composeP{p}{a}) = T_1(3^k + \alpha_{0,-1}(p)) = \frac{3^{k+1} + 3\alpha_{0,-1}(p) +1}{2} = T_{1,k}(\alpha_{0,-1}(p))$ by Appendix~\ref{app:structT}, Lemma~\ref{lem:t1}.
\end{enumerate}
\end{itemize}
In all the cases we get the result.
\end{proof}
\fi
\begin{example}
\if 0\mode On Figure~\ref{fig:order}\else On Figure~\figGene in \cite{Collatz1arxiv}\fi, we are reading $\alpha_{0,-1}(\downarrow\leftarrow\leftarrow) = 8$. On the other hand, Theorem~\ref{th:end} claims that $\alpha_{0,-1}(\downarrow\leftarrow\leftarrow) = T_{1,2}(\alpha_{0,-1}(\downarrow\leftarrow)) = T_{1,2}(2) $. Let's verify that: $T_{1,2}(2) = 2^{-1}_{2}(3*2 + 1) = 3 + 2^{-1}_{2} = 3 + \frac{3^2 + 1}{2} = 3 + 5 = 8$ as expected.
\end{example}
\subsection{The \theTree}\label{sec:theTree}

\begin{figure}[h!]
\centering
        \if 1\mode
        \begin{tikzpicture}[scale=0.53,level 1/.style={sibling distance=5cm},level 2/.style={sibling distance=7cm},
    level 3/.style={sibling distance=3.5cm}]
        \else
        \begin{tikzpicture}[scale=0.8,level 1/.style={sibling distance=5cm},level 2/.style={sibling distance=7cm},
    level 3/.style={sibling distance=3.5cm}]
        \fi
    
      \node {$(\epsilon,0,0)^{*}$}
        child {
          node {$(\leftarrow,2,1)^{\star}$}
          child {
            node {$(\leftarrow\leftarrow,8,2)$}
            child { node {$(\leftarrow\leftarrow\leftarrow,26,3)$} }
            child { node {$(\leftarrow\leftarrow\downarrow,4,2)$} }
          }
          child {node {$(\leftarrow\downarrow,1,1)$}
            child { node {$(\leftarrow\downarrow\leftarrow,2,2)$} }
            child { node {$(\leftarrow\downarrow\downarrow,2,1)^{\star}$} }
          }
        }
        child {node {$(\downarrow,0,0)^{*}$}
            child[missing]
        };
    \end{tikzpicture}
    \caption{\small First 4 levels of the \theTree. Two symmetries are highlighted by $*$ and $\star$.}\label{fig:theTree}
\end{figure}
Theorem~\ref{th:end} implies that the operators $T_{0,k}$ and $T_{1,k}$ naturally give birth to a binary tree ruling the construction of $\alpha_{0,-1}$. We call this tree the \theTree:
\begin{definition}[The \theTree]\label{def:theTree}
We call the \theTree \ the binary tree with nodes in $\Nt \subset \left ( \Pa\times\N\times\N \right )$ constructed as follow, starting from node $x=(\epsilon,0,0)$:
\begin{enumerate}
\if 1\mode
\vspace{-1ex}\fi
\item The right child of $(p,x,k)$ is $((\composeP{p}{\downarrow}),T_{0,k}(x),k)$
\item The left child of $(p,x,k)$ is given by $((\composeP{p}{\leftarrow}),T_{1,k+1}(x),k+1)$
\end{enumerate}
\end{definition}
\begin{lemma}\label{lem:nodesTree}
Nodes of the \theTree\ are: $\Nt = \{(p,\alpha_{0,-1}(p),l(p)) \text{ for } p\in\Pa\}$.
\end{lemma}
\begin{proof} Each node of \theTree\ corresponds to a first occurrence, immediate from Definition~\ref{def:theTree} and Theorem~\ref{th:end}.
\end{proof}
\paragraph*{Symmetries of the \theTree} Figure~\ref{fig:theTree} illustrates the first four levels of the \theTree. By construction of the \theTree, if two nodes $(p,x,k)$ and $(p',x,k)$ share the same $x$ and $k$ they will be the root of very similar sub-trees. This phenomenon is highlighted with the nodes $(\epsilon,0,0)$ and $(\downarrow,0,0)$, Figure~\ref{fig:theTree} doesn't show the sub-tree under $(\downarrow,0,0)$ as it can be entirely deduced from the sub-tree under $(\epsilon,0,0)$. The same would apply for the sub-trees under $(\leftarrow,2,1)$ and $(\leftarrow\downarrow\downarrow,2,1)$. These symmetries are closely related to the fact that keeping adding $\downarrow$ to a parity vector of span $k$ will periodically enumerate $(\ZnZ{3^k})^*$ (Lemma~\ref{lem:down}).
\section{Regular expressions defining ancestors sets}\label{sec:sym}
Pursuing our primary goal, we wish to characterize the binary expression of ancestors of an arbitrary $x$ in the Collatz graph. We decompose the set of all ancestors of $x$ as the union on $k$ of sets $\Pred{k}{x}$. The set $\Pred{k}{x}$ contains all the ancestors of $x$ which use the map $T_1$ exactly $k$ times in order to reach $x$ -- the map $T_0$ can be used an arbitrary number of times.

\begin{remark}\label{rk:justi}
The set $\Predk{x}$ appears naturally in a fast-forwarded version of the Collatz process where even steps are ignored and only odd steps are considered (see \cite{DBLP:journals/tcs/Colussi11}). In the graph of that process, $\Predk{x}$ corresponds to the set of ancestors of $x$ at distance $k$. 
\end{remark}

Let's start by noticing the following:
\begin{restatable}{lemma}{lemmult}
\label{lem:mul3}
Let $x\in\N$. If $x$ is a multiple of 3 then: $\forall k > 0, \; \Predk{x} = \emptyset$.
\end{restatable}
\begin{proof}
More generally, if $x$ is a multiple of three, the set of ancestors of $x$ in the Collatz graph is reduced to 
$\Pred{0}{x} = \{ 2^n x \text{ for } n\in\N \}$. Indeed, $T^{-1}_1(x) = \frac{2x-1}{3}$ cannot be an integer if $\mud{x}{0}{3}$ and $\mud{x}{0}{3} \Rightarrow \forall n\in\N,\; \mud{2^n x}{0}{3}$.
\end{proof}
\begin{remark}
In fact, sets $\Pred{k}{x}$ are infinite for all $k$ as soon as $x$ is not a multiple of $3$.
\end{remark}
Thanks to Section~\ref{sec:first}, we know that we can describe the binary expression of elements of $\Pred{k}{x}$ by focusing on parity vectors: the function $\E$ will translate parity vectors to the binary expressions of ancestors. Let's make that link formal:
\begin{definition}[$\E\Predk{x}$]\label{def:epred}
Let $x\in\N$ and $k\in\N$. We define the set $\E\Predk{x} \subset \Bs$ to be:
$ \E\Predk{x} = \{  \omega \bullet \E(p) \; | \; p\in\Pa \text{ such that } \alpha_{0,-1}(p) = x\text{ \emph{mod} } 3^k \text{ and } l(p) = k  \}$. With $\omega = \eta$ if $x < 3^k$ or $\omega = \I^{-1}(i)$ otherwise, and $i=\lfloor \frac{x}{3^k} \rfloor$. By $x\text{ \emph{mod} } 3^k$, we mean ``the rest in the Euclidean division of $x$ by $3^k$''.
\end{definition}
The set $\E\Predk{x}$ constains binary representations of elements of $\Predk{x}$ -- with potential leading $0$s -- in a one-to-one correspondence:
\begin{restatable}{lemma}{lembijg}
\label{lem:bijg}
Let $x>0$ and $k\in\N$. The sets $\E\Predk{x}$ and $\Predk{x}$ are in bijection by the function $g: \E\Predk{x} \to \Predk{x}$ defined by $g(\omega) = \I(\omega)$.
\end{restatable}
\if 1\mode
\begin{proof}
Straightforward, proof in the expanded version \cite{Collatz1arxiv}.
\end{proof}
\else
\begin{proof}
\begin{enumerate}
\item The function $g$ is well defined. Indeed, for any $\omega\in\E\Predk{x}$ \label{point:lembijg1}
\begin{align*}
\omega \in \E\Predk{x} &\Leftrightarrow \exists p \in \Pa \; \I(\omega) = 2^{\norm{p}}i + \I(\E(p)) \text{ with } \alpha_{0,-1}(p) = x \text{ mod } 3^k \text{ and } l(p) = k\\
& \Leftrightarrow \I(\omega) = 2^{\norm{p}}i + \alpha_{0,0}(p) \\
& \Leftrightarrow \I(\omega) = \alpha_{i,0}(p) = g(\omega) \in \Predk{x}
\end{align*}
With $i = \lfloor \frac{x}{3^k} \rfloor$.
\item The function $g$ is injective. Let $\omega_1,\omega_2\in\E\Predk{x}$ with $\omega_1 = \omega\bullet\E(p_1)$ and $\omega_2 = \omega\bullet\E(p_2)$ with $\omega=\eta$ if $i=\lfloor \frac{x}{3^k} \rfloor=0$ else $\omega=\I^{-1}(i)$. We have $x=\alpha_{i,-1}(p_1) = \alpha_{i,-1}(p_2)$ by hypothesis. Suppose $g(\omega_1) = g(\omega_2)$. We get $2^{\norm{p_1}}i + \I(\E(p_1)) = 2^{\norm{p_2}}i + \I(\E(p_2))$. Hence, $2^{\norm{p_1}}i + \alpha_{0,0}(p_1) = 2^{\norm{p_2}}i + \alpha_{0,0}(p_2)$. By Theorem~\ref{th:struct_occ} we get $\alpha_{i,0}(p_1) = \alpha_{i,0}(p_2)$. If $\norm{p_1} \neq \norm{p_2}$, for instance $\norm{p_1} < \norm{p_2}$ we have $p_2 = \composeP{p_1}{(\downarrow)^{\norm{p_2}-\norm{p_1}}}$. Indeed, by determinism of the Collatz process, $p_1$ must be a prefix of $p_2$ as they both are forward feasible for $y = \alpha_{i,0}(p_1) = \alpha_{i,0}(p_2)$. Furthermore, we can't add any more arrows of type $\leftarrow$ because $l(p_1)=l(p_2)$. But, $\alpha_{i,-1}(p_1) = x \neq 0$ thus $\alpha_{i,-1}(p_2) = x/(2^{\norm{p_2}-\norm{p_1}}) \neq x$ which contradicts $\alpha_{i,-1}(p_1)=\alpha_{i,-1}(p_2)$. Hence we have $\norm{p_1}=\norm{p_2}$ and thus $p_1=p_2$ because, by determinism of the Collatz process, there is only one path of a given norm between $\alpha_{i,0}(p_1)$ and $\alpha_{i,-1}(p_1)$ and thus one corresponding parity vector. Hence, $\omega_1 = \omega_2$.
\item The function $g$ is surjective. Let $y\in\Predk{x}$. There exists $p\in\Pa$ with $\alpha_{i,0} = y$ and $\alpha_{i,-1} = x$ with $i = \lfloor \frac{x}{3^k} \rfloor$. Similarly to the proof of Point~\ref{point:lembijg1}, the reader can verify that $\E(p)\in\E\Predk{x}$ is a valid antecedent of $y$ in the case $i=0$ and that $\I^{-1}(i)\bullet\E(p)\in\E\Predk{x}$ is a valid antecedent of $y$ otherwise.
\end{enumerate}
\end{proof}
\fi
Hence, in order to describe $\E\Predk{x}$ we are concerned by characterizing parity vectors $p$ such that $\alpha_{0,-1}(p) = x \text{ mod } 3^k$ and $l(p) = k$. Such $p$ correspond to the symmetries that we highlighted in the \theTree, they form an equivalence class of ``$k$-span equivalence'':
\begin{definition}[$k$-span equivalence]\label{def:simeq}
Two parity vectors $p_1,p_2\in\Pa$ are said to be $k$-span equivalent if $l(p_1)=l(p_2)=k$ and $\alpha_{0,-1}(p_1)=\alpha_{0,-1}(p_2)$. We write $p_1\simeq_k p_2$. Note that $\simeq_k$ is an equivalence relation. 
\end{definition}
The following set of binary strings will play a central role in how we can describe $k$-span equivalence classes:
\begin{definition}[Parity sequence of $(\ZnZ{3^k})^{*}$]\label{def:pik}
For $k>0$, we define $\Pi_k \in \Bs$, the parity sequence of $(\ZnZ{3^k})^{*}$ as follows:
    $\Pi_k = b_0 \dots b_{\pi_k-1}$ with $|\Pi_k| = \pi_k = |(\ZnZ{3^k})^{*}| = 2*3^{k-1}$  and, 
    $b_{\pi_k-1-i} = \textup{\texttt{0}}$ if  $2^{-i}_k$  is even and $b_{\pi_k-1-i} = \textup{\texttt{1}}$ if $2^{-i}_k$is odd.
By convention, we fix $\pi_0 = 1$.
\end{definition}
\begin{example}
For $k=3$, we have $2^{-1}_3 = 14$. The sequence of powers of $2^{-1}_3$ in $(\ZnZ{3^k})^{*}$ is: $[1, 14, 7, 17, 22, 11, 19, 23, 25, 26, 13, 20, 10, 5, 16, 8, 4, 2]$. The associated parity sequence ($\textup{\texttt{0}}$ when even and $\textup{\texttt{1}}$ when odd) is: $\textup{\texttt{101101111010010000}}$. Finally, $\Pi_3$ is the mirror image of this: $\Pi_3 = \textup{\texttt{000010010111101101}}$. We have: $
\Pi_1 = \textup{\texttt{01}}$, $\Pi_2 = \textup{\texttt{000111}}$, $\Pi_3 = \textup{\texttt{000010010111101101}}$ and\\ $\Pi_4 = \textup{\texttt{000000110010100100010110000111111001101011011101001111}}$.
\end{example}
\begin{remark}
The strings $\Pi_k$, or ``seeds'' in \cite{DBLP:journals/tcs/Colussi11}, have been studied in great depth in \cite{complexity13}.
The author find that their structure is extremely complex, that they have numerous properties and that they can be defined in a lot of different ways.
For instance, \cite{DBLP:journals/tcs/Hew16} uses the fact that strings $\Pi_k$ correspond to the repetend of $1/3^k$ in binary. 
\end{remark}
\begin{definition}[Rotation operator $\R{i}{\cdot}$]
Let $\omega \in \Bs$ with $|\omega| = n$. Then, for $0 \leq i < n$, $\R{i}{\omega}$ denotes the $i^\text{th}$ rotation (or circular shift) to the right of $\omega$. For instance, we have $\R{2}{\textup{\texttt{000111}}} = \textup{\texttt{110001}}$.
\end{definition}

\begin{figure}[t!]

\centering

\begin{tikzpicture}[scale=0.30]

\begin{scope}
    \draw [fill] (0, 0) circle [radius=0.1];
    \draw [->] [thick] [brown] (0,1) -- (0,0);
    \draw [fill] (0, 1) circle [radius=0.1];
    \draw [->] [thick] [brown](0,2) -- (0,1);
    \draw [fill] (0, 2) circle [radius=0.1];
    \draw [->] [thick] [brown](0,3) -- (0,2);
    \draw [fill] (0, 3) circle [radius=0.1];
    \draw [->] [thick] [brown](0,4) -- (0,3);
    \draw [fill] (0, 4) circle [radius=0.1];
    \draw [->] [thick] [brown](0,5) -- (0,4);
    \draw [fill] (0, 5) circle [radius=0.1];
    \draw [->] [thick] [brown](0,6) -- (0,5);
    \draw [fill] (0, 6) circle [radius=0.1];

    \draw [dashed] [->] (0,7) -- (0,6);
    \draw [dashed] [->] (0,8) -- (0,7);

    \draw [thick] [blue] [->] (1,0) -- (0,0);
    \draw [fill] (1, 0) circle [radius=0.1];
    \draw [thick] [blue] [->] (2,0) -- (1,0);
    \draw [fill] (2, 0) circle [radius=0.1];
    \draw [thick] [blue] [->] (2,1) -- (2,0);
    \draw [fill] (2, 1) circle [radius=0.1];

    \begin{scope}[yshift=6cm]
        \draw [thick] [brown] [->] (1,0) -- (0,0);
        \draw [fill] (1, 0) circle [radius=0.1];
        \draw [thick] [brown] [->] (2,0) -- (1,0);
        \draw [fill] (2, 0) circle [radius=0.1];
        \draw [thick] [brown] [->] (2,1) -- (2,0);
        \draw [fill] (2, 1) circle [radius=0.1];
    \end{scope}

    \node [below] at (0,0) {$8$};
    \node [below] at (0,-1) {$2^{-3}_2$};
\end{scope}

\begin{scope}[xshift=3cm]
    \draw [fill] (0, 0) circle [radius=0.1];
    \draw [->] (0,1) -- (0,0);
    \draw [fill] (0, 1) circle [radius=0.1];
    \draw [->] (0,2) -- (0,1);
    \draw [fill] (0, 2) circle [radius=0.1];
    \draw [->] (0,3) -- (0,2);
    \draw [fill] (0, 3) circle [radius=0.1];
    \draw [->] (0,4) -- (0,3);
    \draw [fill] (0, 4) circle [radius=0.1];
    \draw [->] (0,5) -- (0,4);
    \draw [fill] (0, 5) circle [radius=0.1];
    \draw [->] (0,6) -- (0,5);
    \draw [fill] (0, 6) circle [radius=0.1];
    \draw [dashed] [->] (0,7) -- (0,6);
    \draw [dashed] [->] (0,8) -- (0,7);
    \begin{scope}[yshift=1cm]
        \draw [thick] [->] (1,0) -- (0,0);
        \draw [fill] (1, 0) circle [radius=0.1];
        \draw [thick] [->] (2,0) -- (1,0);
        \draw [fill] (2, 0) circle [radius=0.1];
        \draw [thick] [->] (2,1) -- (2,0);
        \draw [fill] (2, 1) circle [radius=0.1];
    \end{scope}

    \begin{scope}[yshift=7cm]
        \draw [dashed] [->] (1,0) -- (0,0);
        \draw [dashed] [->] (2,0) -- (1,0);
        \draw [dashed] [->] (2,1) -- (2,0);
    \end{scope}

    \node [below] at (0,0) {$4$};
    \node [below] at (0,-1) {$2^{-4}_2$};
\end{scope}

\begin{scope}[xshift=6cm]
    \draw [fill] (0, 0) circle [radius=0.1];
    \draw [->] (0,1) -- (0,0);
    \draw [fill] (0, 1) circle [radius=0.1];
    \draw [->] (0,2) -- (0,1);
    \draw [fill] (0, 2) circle [radius=0.1];
    \draw [->] (0,3) -- (0,2);
    \draw [fill] (0, 3) circle [radius=0.1];
    \draw [->] (0,4) -- (0,3);
    \draw [fill] (0, 4) circle [radius=0.1];
    \draw [->] (0,5) -- (0,4);
    \draw [fill] (0, 5) circle [radius=0.1];
    \draw [->] (0,6) -- (0,5);
    \draw [fill] (0, 6) circle [radius=0.1];
    \draw [dashed] [->] (0,7) -- (0,6);
    \draw [dashed] [->] (0,8) -- (0,7);
    \begin{scope}[yshift=2cm]
        \draw [thick] [->] (1,0) -- (0,0);
        \draw [fill] (1, 0) circle [radius=0.1];
        \draw [thick] [->] (2,0) -- (1,0);
        \draw [fill] (2, 0) circle [radius=0.1];
        \draw [thick] [->] (2,1) -- (2,0);
        \draw [fill] (2, 1) circle [radius=0.1];
    \end{scope}

    \begin{scope}[yshift=8cm]
        \draw [dashed] [->] (1,0) -- (0,0);
        \draw [dashed] [->] (2,0) -- (1,0);
        \draw [dashed] [->] (2,1) -- (2,0);
    \end{scope}

    \node [below] at (0,0) {$2$};
    \node [below] at (0,-1) {$2^{-5}_2$};
\end{scope}

\begin{scope}[xshift=9cm]
    \draw [fill] (0, 0) circle [radius=0.1];
    \draw [->] (0,1) -- (0,0);
    \draw [fill] (0, 1) circle [radius=0.1];
    \draw [->] (0,2) -- (0,1);
    \draw [fill] (0, 2) circle [radius=0.1];
    \draw [->] (0,3) -- (0,2);
    \draw [fill] (0, 3) circle [radius=0.1];
    \draw [->] (0,4) -- (0,3);
    \draw [fill] (0, 4) circle [radius=0.1];
    \draw [->] (0,5) -- (0,4);
    \draw [fill] (0, 5) circle [radius=0.1];
    \draw [->] (0,6) -- (0,5);
    \draw [fill] (0, 6) circle [radius=0.1];
    \draw [dashed] [->] (0,7) -- (0,6);
    \draw [dashed] [->] (0,8) -- (0,7);
    \begin{scope}[yshift=3cm]
        \draw [thick] [->] (1,0) -- (0,0);
        \draw [fill] (1, 0) circle [radius=0.1];
        \draw [thick] [->] (2,0) -- (1,0);
        \draw [fill] (2, 0) circle [radius=0.1];
        \draw [thick] [->] (2,1) -- (2,0);
        \draw [fill] (2, 1) circle [radius=0.1];
    \end{scope}

    \node [below] at (0,0) {$1$};
    \node [below] at (0,-1) {$2^{-0}_2$};
\end{scope}

\begin{scope}[xshift=12cm]
    \draw [fill] (0, 0) circle [radius=0.1];
    \draw [->] (0,1) -- (0,0);
    \draw [fill] (0, 1) circle [radius=0.1];
    \draw [->] (0,2) -- (0,1);
    \draw [fill] (0, 2) circle [radius=0.1];
    \draw [->] (0,3) -- (0,2);
    \draw [fill] (0, 3) circle [radius=0.1];
    \draw [->] (0,4) -- (0,3);
    \draw [fill] (0, 4) circle [radius=0.1];
    \draw [->] (0,5) -- (0,4);
    \draw [fill] (0, 5) circle [radius=0.1];
    \draw [->] (0,6) -- (0,5);
    \draw [fill] (0, 6) circle [radius=0.1];
    \draw [dashed] [->] (0,7) -- (0,6);
    \draw [dashed] [->] (0,8) -- (0,7);
    \begin{scope}[yshift=4cm]
        \draw [thick] [->] (1,0) -- (0,0);
        \draw [fill] (1, 0) circle [radius=0.1];
        \draw [thick] [->] (2,0) -- (1,0);
        \draw [fill] (2, 0) circle [radius=0.1];
        \draw [thick] [->] (2,1) -- (2,0);
        \draw [fill] (2, 1) circle [radius=0.1];
    \end{scope}
    \node [below] at (0,0) {$5$};
    \node [below] at (0,-1) {$2^{-1}_2$};
\end{scope}

\begin{scope}[xshift=15cm]
    \draw [fill] (0, 0) circle [radius=0.1];
    \draw [->] (0,1) -- (0,0);
    \draw [fill] (0, 1) circle [radius=0.1];
    \draw [->] (0,2) -- (0,1);
    \draw [fill] (0, 2) circle [radius=0.1];
    \draw [->] (0,3) -- (0,2);
    \draw [fill] (0, 3) circle [radius=0.1];
    \draw [->] (0,4) -- (0,3);
    \draw [fill] (0, 4) circle [radius=0.1];
    \draw [->] (0,5) -- (0,4);
    \draw [fill] (0, 5) circle [radius=0.1];
    \draw [->] (0,6) -- (0,5);
    \draw [fill] (0, 6) circle [radius=0.1];
    \draw [dashed] [->] (0,7) -- (0,6);
    \draw [dashed] [->] (0,8) -- (0,7);
    \begin{scope}[yshift=5cm]
        \draw [thick] [->] (1,0) -- (0,0);
        \draw [fill] (1, 0) circle [radius=0.1];
        \draw [thick] [->] (2,0) -- (1,0);
        \draw [fill] (2, 0) circle [radius=0.1];
        \draw [thick] [->] (2,1) -- (2,0);
        \draw [fill] (2, 1) circle [radius=0.1];
    \end{scope}

    \node [below] at (0,0) {$7$};
    \node [below] at (0,-1) {$2^{-2}_2$};
\end{scope}

\end{tikzpicture}

\caption{\small Illustration of Lemma~\ref{lem:down}. How the parity vector $p=\downarrow\leftarrow\leftarrow$ (in blue), with $l(p)=2$, distributes on the elements of $(\ZnZ{3^{2}})^{*}$. The first of occurrence of $p$ is such that $\alpha_{0,-1}=8=2_{2}^{-i_0}=2_{2}^{-3}$. The parity vector $p$ is $k$-span equivalent to the parity vector $p'=\downarrow\leftarrow\leftarrow(\downarrow)^6$ (in brown).}\label{fig:poles}

\end{figure}

From any parity vector $p$, we can create an infinite family of distinct parity vectors which are $k$-span equivalent to $p$:
\begin{restatable}{lemma}{lemdown}
\label{lem:down}
Let $p\in\Pa$ and $k=l(p) > 0$. Define $p_n = \composeP{p}{(\downarrow)^{n\pi_k}}$, i.e. the parity vector $p$ followed by $n\pi_k$ arrows of type $\downarrow$,  where $\pi_k = |\Pi_k|$. Then, for all $n\in\N$ we have $p \simeq_k p_n $. Furthermore we can characterize $\alpha_{0,0}(p_n)$ through $\E(p_{n})$ with:
$$\E(p_{n+1}) = \R{i_0}{\Pi_k} \bullet \E(p_{n}) \Leftrightarrow \E(p_{n}) = (\R{i_0}{\Pi_k})^n \bullet \E(p) 
$$
With $0 \leq i_0 < \pi_k$ such that $\alpha_{0,-1}(p) = 2^{-i_0}_{k}$ in $(\ZnZ{3^k})^{*}$.
\end{restatable}
\if 1\mode
\begin{proof}
\textbf{(Sketch)} Direct consequence of Theorem~\ref{th:begin} and~\ref{th:end}. Full proof in the expanded version \cite{Collatz1arxiv}.
\end{proof}
\else
\begin{proof}

We have $l(p_n) = l(p) = k > 0$. By Lemma~\ref{lem:star}, we know that $\alpha_{0,-1}(p)\in(\ZnZ{3^k})^*$. Furthermore, by Theorem~\ref{th:end}, $\alpha_{0,-1}(p_n) = T^{n\pi_k}_{0,k}(\alpha_{0,-1}(p)) = (2^{-n}_k)^{\pi_k} \alpha_{0,-1}(p) =  1*\alpha_{0,-1}(p) = \alpha_{0,-1}(p)$ since $\pi_k$ is the order of the group $(\ZnZ{3^k})^*$. Hence we have $p_n \simeq_k p$. Furthermore, by Theorem~\ref{th:begin} we know that $\E(p_{n+1}) = \omega \bullet \E(p_{n})$ with $\omega=b_0\ldots b_{\pi_k - 1} \in\Bs$ with $|\omega| = \pi_k$ such that:
$$ b_{\pi_k - 1 - i } = \begin{cases} \texttt{0} &\text{ if } 2^{-i_0 - i}_k \text{ is even}\\\texttt{1} &\text{ if } 2^{-i_0 - i}_k \text{ is odd}\end{cases}$$
With $0 \leq i_0 < \pi_k$ such that $\alpha_{0,-1}(p) = 2^{-i_0}_{k}$ in $(\ZnZ{3^k})^{*}$. By definition, the string $\omega=b_0\ldots b_{\pi_k - 1}$ is exactly $\R{i_0}{\Pi_k}$ and we have the result.

\end{proof}
\fi
\begin{remark}
The result of Lemma~\ref{lem:down} is illustrated in Figure~\ref{fig:poles}. The parity vector $p$ (in blue) distributes in a ``spiral" around the elements of $(\ZnZ{3^k})^*$. When $\pi_k$ arrows of type $\downarrow$ have been added to $p$, a full ``turn" has been done and we get a path $k$-span equivalent to $p$. As a consequence, following only right children in the \theTree\ exhibits periods of length $\pi_k$ which enumerate elements of $(\ZnZ{3^k})^*$.
\end{remark}
We now have all the element in order to characterize $\E\Pred{k}{x}$ using regular expressions:
\thmain*

\begin{proof}We are going to explicitly construct $\regk{x}$, a regular expression\footnote{The regular expressions we work with are defined by the following BNF: $$\text{reg} := \emptyset  \; | \; (\omega \in \Bs) \; | \; (\text{reg}_1|\text{reg}_2) \; | \; (\text{reg})^* \; | \; (\text{reg}_1)(\text{reg}_2)$$  For instance, the expression $(\texttt{01})^*((00)|(11))$ matches any word of the form $(\texttt{01})^n \texttt{00}$ or $(\texttt{01})^n \texttt{11}$. We might omit some parenthesis when they are redundant.} which defines $\E\Predk{x}$. With the following preliminary argument we show that it is enough to construct $\regk{x}$ when $x$ is not a multiple of 3 and $x < 3^k$. In other words, when $x\in(\ZnZ{3^k})^*$.

\textbf{Preliminary Argument.} Let $x,k\in\N$. Suppose $x$ is a multiple of 3. If $k>0$, by Lemma~\ref{lem:mul3}, $\E\Predk x$ is empty and thus we can take $\reg{k}{x} = \emptyset$ in that case. If $k=0$, $\E\Pred{0}{x} = \{ \omega \bullet (\texttt{0})^n \text{ for } n\in\N \}$ with $\omega=\eta$ if $x=0$ or $\omega = \I^{-1}(x)$ otherwise (Definition~\ref{def:epred} and Theorem~\ref{th:begin}). Hence we take $\reg{0}{0}=(\texttt{0})^*$ and $\reg{0}{x}= (\I^{-1}(x))(\texttt{0})^*$ for $x>0$.

Suppose $x$ is not a multiple of 3 and $x\geq3^k$. Suppose that $\reg{k}{(x \text{ mod } 3^k)}$ exists, i.e. that the set $\E\Predk{x'}$ is regular with $x' = (x \text{ mod } 3^k)$. Then by Definition~\ref{def:epred} we can take $\reg{k}{x} = (\I^{-1}(\lfloor \frac{x}{3^k} \rfloor ))(\reg{k}{(x \text{ mod } 3^k)})$ in order to define $\E\Predk x$. Indeed, by Theorem~\ref{th:struct_occ}, $\{p\in\Pa\,|\, \alpha_{i,-1}(p) = x\} = \{p\in\Pa\,|\, \alpha_{0,-1} = (x \text{ mod } 3^k)\}$ with $i=\lfloor \frac{x}{3^k} \rfloor$.

Hence we just have to prove that $\regk{x}$ exists for all $x$, non multiple of three such that $x < 3^k$, i.e. $x\in(\ZnZ{3^k})^{*}$. We prove by induction on $k$ the following result:
$$ H(k) = `` \forall x\in(\ZnZ{3^k})^* \; \text{ there exists } \regk{x} \text{ which defines } \E\Predk x" $$ 
\textbf{Induction.}\\
\textbf{Base step $k=0$.} Trivially true because $(\ZnZ{3^k})^* = \emptyset$. Note that the following induction step will rely on knowing $\reg{0}{0}$. We have shown above that $\reg{0}{0} = (0)^{*}$.

\noindent \textbf{Inductive step.} Let $k\in\N$ such that $H(k)$ holds. We show that $H(k+1)$ holds. Let $x\in(\ZnZ{3^{k+1}})^*$ and $0 \leq i_0 < \pi_{k+1}$ such that $x = 2^{-i_0}_{k+1}$. By Definition~\ref{def:epred}, in this case, we have $\E\Pred{k+1}{x} = \{ \E(p) \; | \; p\in\Pa \text{ such that } \alpha_{0,-1}(p) = x \text{ and } l(p) = k+1  \}$. Hence, characterizing $\E\Pred{k+1}{x}$ boils down to characterizing the $(k+1)$-span equivalence class: $\{ p \; | \; p\in\Pa \text{ such that } \alpha_{0,-1}(p) = x \text{ and } l(p) = k+1  \} = \{ p \; | \; (p,x,k+1)\in\Nt \} $  (Lemma~\ref{lem:nodesTree}).

Hence, we take $p$ such that $(p,x,k+1)$ is in the \theTree\ and we analyse its structure. To do so, we consider the surrounding of $p$ in the \theTree. This will lead us to Equation~\eqref{eq:mainproot} which relates $\E(p)$ to the induction hypothesis. We are going to deploy Points 1, 2, 3, in order to show that the node $(p,x,k+1)$ can always be expressed in the context of Figure~\ref{fig:illus}:

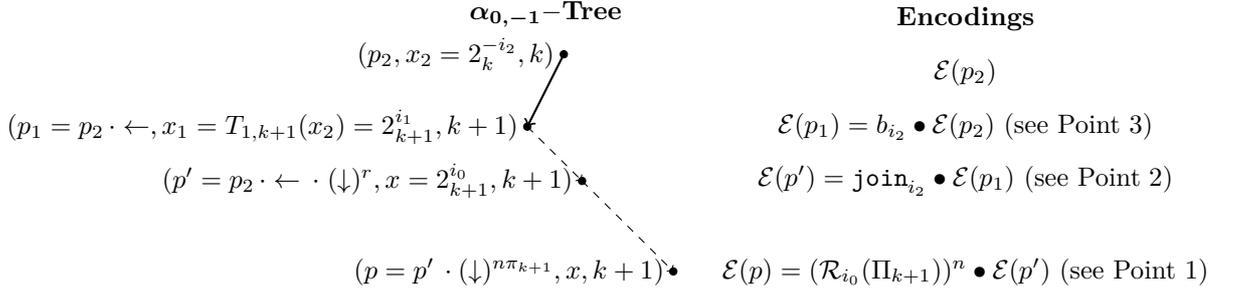
\begin{figure}[h!]
\centering
    \begin{tikzpicture}[scale=0.48]
    \begin{scope}[xshift=-29cm]
    \node [above] at (-0.5, 0.5) {$\boldsymbol{\alpha_{0,-1}}-$\textbf{Tree}};
        \node [left] at (0, 0) {$(p_2,x_2 = 2^{-i_2}_k, k)$};
        \draw [fill] (0,0) circle [radius=0.1];
        \draw [thick] [->] (0,0) -- (-1,-2);
        \draw [fill] (-1,-2) circle [radius=0.1];
        \node [left] at (-1,-2) {$(p_1=\composeP{p_2}{\leftarrow},x_1 = T_{1,k+1}(x_2) = 2^{i_1}_{k+1}, k+1)$};
        \draw [dashed] [->] (-1, -2) -- (0.5,-3.5);
        \draw [fill] (0.5,-3.5) circle [radius=0.1];
        \node [left] at (0.5,-3.5) {$(p'=\composeP{\composeP{p_2}{\leftarrow}}{\;(\downarrow)^r}, x = 2^{i_0}_{k+1}, k+1)$};
        \draw [dashed] [->] (0.5,-3.5) -- (3,-6);
        \node [left] at (3,-6) {$(p=\composeP{p'}{(\downarrow)^{n \pi_{k+1}}}, x, k+1)$};
        \draw [fill] (3,-6) circle [radius=0.1];
    \end{scope}
    \begin{scope}[xshift=-18cm]
     \node [above] at (0, 0.4) {\textbf{Encodings}};
     \node  at (0, -0.5) {$\E(p_2)$ };
     \node  at (0, -2) {$\E(p_1) = b_{i_2} \bullet \E(p_2) $ (see Point 3)};
     \node  at (0,-3.5) {$\E(p') = \join_{i_2} \bullet \E(p_1) $ (see Point 2)};
     \node  at (0,-6) {$\E(p) = (\R{i_0}{\Pi_{k+1}})^n \bullet \E(p') $ (see Point 1)};
    \end{scope}

     




    \end{tikzpicture}
    \caption{Situation of the node $(p,x,k+1)$ is the \theTree. Each black dot is a node in the \theTree. The solid edge reaches one left child while dashed edges represent variable numbers of right children (see Definition~\ref{def:theTree}). This Figure is to be read bottom to top together with Points 1,2,3.}\label{fig:illus}
\end{figure}

In order to prove the generality of this situation, three points:
\begin{enumerate}
    \item Since $l(p) = k +1\geq 1$ we can decompose $p = p_2\;\cdot\leftarrow\cdot\;(\downarrow)^m$ with $m\in\N$ and $p_2\in\Pa$ such that $l(p_2) = k$. We can write $m = n\pi_{k+1} + r$ with $r < \pi_{k+1}$ and $p = p_2\;\cdot\leftarrow\cdot\;(\downarrow)^r\cdot(\downarrow)^{n\pi_{k+1}}$. We call the number $n$ the \textit{repeating value}. By Lemma~\ref{lem:down}, we know that $p \simeq_{k+1} \composeP{\composeP{p_2}{\leftarrow}}{\;(\downarrow)^r}$. Hence, with $p'=\composeP{\composeP{p_2}{\leftarrow}}{\;(\downarrow)^r}$, we have $\alpha_{0,-1}(p')=\alpha_{0,-1}(p)=x$ and $\E(p) = (\R{i_0}{\Pi_{k+1}})^n \bullet \E(p') $. It remains to characterize $\E(p') = \E(\composeP{\composeP{p_2}{\leftarrow}}{\;(\downarrow)^r}) = \E(\composeP{p_1}{\;(\downarrow)^r})$ with $p_1 = \composeP{p_2}{\leftarrow}$.

\item Let's consider now $x_1 = \alpha_{0,-1}(\composeP{p_2}{\leftarrow})$. By Theorem~\ref{th:end}, we have $x_1 = T_{1,k+1}(x_2) = T_{1,k+1}(2^{-i_2}_k)$. Furthermore, by the same Theorem~\ref{th:end}, we have $x = \alpha_{0,-1}(p') = \alpha_{0,-1}(\composeP{p_1}{\downarrow^r}) = T_{0,k+1}^{r}(\alpha_{0,-1}(\composeP{p_2}{\leftarrow})) = T_{0,k+1}^{r}(x_1)$. Hence $x = 2^{-r}_{k+1} x_1 $ and thus $x_1 = 2^{-i_1}_{k+1}$ with $0 \leq i_1 = -i_0 + r < \pi_{k+1}$. By Theorem~\ref{th:begin}, we deduce that: $ \E(p' = \composeP{p_1}{\;(\downarrow)^r}) = \omega \bullet \E(p_1) $.
With $\omega = j_0 \ldots j_{r-1} \in \Bs$, $|\omega| = r$, and $j_{r-i} = \begin{cases} \texttt{0} \text{ if } 2^{-i_1 - i}_{k+1} \text{ is even} \\  \texttt{1} \text{ if } 2^{-i_1 - i}_{k+1} \text{ is odd}  \end{cases}$ with $0 \leq i < r$. 

We refer to such $\omega$ by $\join_{i_2}$ because\footnote{The name $\texttt{join}$ refers to the fact that this parity sequence $\omega$ arises from ``joining'', in the \theTree, $x_1 = 2^{i_1}_{k+1}$ to $x = 2^{i_0}_{k+1}$ with $r = i_0-i_1 \geq 0$ arrows of type $\downarrow$. Notice that we can have $\join_{i_2} = \eta$ in the case where $r=0$. } it is uniquely determined by $i_2$ such that $x_1 = T_{1,k+1}(x_2) = T_{1,k+1}(2^{-i_2}_k)$. Indeed, \if 0\mode Lemma~\ref{lem:t1} \else Lemma~\lemTTT of \cite{Collatz1arxiv} \fi shows that $T_{1,k+1}$ is injective on $\ZnZ{3^k}$ hence $i_2\neq i'_2 \Rightarrow T_{1,k+1}(2^{-i_2}_k) \neq T_{1,k+1}(2^{-i'_2}_k)$. Different values of $i_2$ will yield to different $x_1$ and thus different $i_1$, $r$ and $\join_{i_2}$. 

\item Let's consider $x_2 = \alpha_{0,-1}(p_2)$. If $k\neq0$, by Lemma~\ref{lem:star}, we know that $x_2 \in (\ZnZ{3^k})^*$ and we can write $x_2 = 2_k^{-i_2}$ with $0 \leq i_2 < \pi_k$. Note that if $k=0$, by the same Lemma, we have $x_2=0=2^{0}_0$ by convention. Thus in all case we can write $x_2 = 2_k^{-i_2}$ with $0 \leq i_2 < \pi_k$. By Theorem~\ref{th:begin}, we deduce that $\E(p_1) = b_{i_2} \bullet \E(p_2)$ where $b_{i_2} = \texttt{0}$ if $\leftarrow$ is admissble for $x_2$ and $b_{i_2} = \texttt{1}$ otherwise. In other words: $b_{i_2} = \begin{cases} \texttt{0} \text{ if } 2_k^{-i_2} \texttt{ is odd}\\ \texttt{1} \text{ if } 2_k^{-i_2} \texttt{ is even}\end{cases}$.

\end{enumerate}

Over all, from Points 1, 2, 3, we deduce that:
\begin{align}\E(p) = (\R{i_0}{\Pi_{k+1}})^n \bullet  (\join_{i_2}) \bullet (b_{i_2}) \bullet \E(p_2) \label{eq:mainproot}\end{align}

We have $l(p_2) = k$. If $k\neq0$, we will be able to reduce to the induction hypothesis since $x_2\in(\ZnZ{3^k})^*$. If $k=0$ we have $x_2=0$ and we use $\reg{0}{0}$ previously constructed.

As a synthesis, notice that any value of $0 \leq i_2 < \pi_k$, any node $(p_2,x_2=2^{i_2}_k,k)$ and any repeating value $n\in\N$ will lead to the construction of a $(p,x,k+1)$ with a different $p$ for each choice of $i_2$, $p_2$ and $n$. Hence, we have completely characterized the structure of nodes of the form $(p,x,k+1)$. From the above analysis, we can deduce the recursive expression of $\reg{k+1}{x}$, we have:
\begin{align}
\begin{split}
\reg{k+1}{x} = (\R{i_0}{\Pi_{k+1}})^*( \quad &(\join_0) \; (b_0) \; (\reg{k}{2^{-0}_{k}}) \quad |\\
                                             &(\join_1) \; (b_1) \; (\reg{k}{2^{-1}_{k}}) \quad |\\
                                             &(\join_2) \; (b_2) \; (\reg{k}{2^{-2}_{k}}) \quad |\\
                                             &\quad\quad\vdots\\
                                             &(\join_{\pi_k-1}) \; (b_{\pi_k-1}) \; (\reg{k}{2^{-(\pi_k-1)}_{k}}) \quad )
\label{eq:mother}
\end{split}
\end{align}

Note the amusing fact that for $k>0$ the word $b_{\pi_k-1} b_{\pi_k-2}\ldots b_{0}$ is the binary complement of $\Pi_k$. The fact that $\reg{k}{x}$ is structured as a tree is made obvious by Equation~\eqref{eq:mother}, at each level $l \leq k$ the branching factor is $\pi_l$. The number of branches is given by $\prod_{l=0}^{k} \pi_k = 2^k 3^\frac{k(k-1)}{2}$. The number of $\{0,1\}$ symbols on each branch is bounded by $2\sum_{l=0}^k \pi_k = O(3^k)$, hence, the alphabetic width of $\reg{k}{x}$ which is the total number of $\{0,1\}$ symbols is $O(2^k 3^\frac{k(k+1)}{2})$. Finally, the star heightt\footnote{The star height metric \texttt{height} is defined by $\texttt{height}(s) = 0$ for $s\in\{\emptyset,\epsilon,0,1\}$, $\texttt{height}(s_1 | s_2) = \texttt{height}(s_1s_2) = \max(\texttt{height}(s_1), \texttt{height}(s_2))$ and $\texttt{height}(s^*) = \texttt{height}(s)+1$.}\cite{surveyRegexAuto} is $1$ as directly derived from \eqref{eq:mother}.
\end{proof}

\begin{remark}
We remark that the number of branches of $\regk{x}$ corresponds to the size of ``Level $k+1$'' that was computed in \cite{1504.03040}. The author of \cite{1504.03040} also remarks that this number corresponds to the number of different antisymmetric binary relations on a set of $k+1$ labeled points \cite{oeisA083667}.
\end{remark}

\begin{remark}
The Collatz conjecture is equivalent to: for all $x$, there is $k$ such that $x\in \Pred{k}{1}$. Which means: for all $x$, there is $k$ and $n$ such that $\reg{k}{1}$ matches $0^n \mathcal{I}^{-1}(x)$. Because the number of leading $0$s in $\Pi_k$ is equal to $\lfloor k \cdot \text{ln}(3)/\text{ln}(2) \rfloor$ (see \cite{complexity13}), we can bound the number of leading $0$s that is accepted by $\reg{k}{1}$ and so we can bound $n$ which is the number $0$s to prepend to $\mathcal{I}^{-1}(x)$. For instance $n < (k+1)^2$ is a bound that works.
\end{remark}

\if 0\mode
\begin{example}
\label{ex:final}
We implemented the construction of $\reg{k}{x}$ in a Python library named \texttt{coreli} (see Appendix~\ref{app:code}).
Let's consider $\reg{3}{14}$ which defines the binary expression of any $y$ that reaches $14$ in $3$ odd steps. Below are enumerated the 12 branches of $\reg{3}{14}$. They all start with the term $(100001001011110110)^*$ (corresponds to $\R{i_0}{\Pi_{3}}^*$ in the proof of Theorem~\ref{th:main}) which is omitted for readibility:\\
\scalebox{1}{\parbox{.5\linewidth}{%
\begin{align*}
1.\;&100(000111)^*\boldsymbol{00}(01)^*01(0)^* &7.\;& 100(000111)^*\boldsymbol{0001}(10)^*1(0)^* \\
2.\;&10000(100011)^*\boldsymbol{100}(01)^*01(0)^* &8.\;&10000(100011)^*\boldsymbol{10001}(10)^*1(0)^* \\
3.\;&100001001011110(110001)^*\boldsymbol{1100}(01)^*01(0)^*&9.\;& 100001001011110(110001)^*\boldsymbol{110001}(10)^*1(0)^*\\
4.\;&10000100101(111000)^*\boldsymbol{11100}(01)^*01(0)^* &10.\;& 10000100101(111000)^*\boldsymbol{1}(10)^*1(0)^*\\
5.\;&100001001(011100)^*\boldsymbol{011100}(01)^*01(0)^* &11.\;&100001001(011100)^*\boldsymbol{01}(10)^*1(0)^*\\
6.\;&10000100101111011(001110)^*\boldsymbol{0}(01)^*01(0)^* &12.\;& 10000100101111011(001110)^*\boldsymbol{001}(10)^*1(0)^*
\end{align*}
}}

Any number of which binary expression matches one of the 12 branch will iterate, in $3$ odd steps, to $\boldsymbol{1110}$, the binary representation of $14$. One can notice that both columns are enumerating, in the same order, all rotations of the word $(000111)$ in the left-most Kleene star. First column corresponds to rotation $(01)$ of the second Kleene star and second column corresponds to rotation $(10)$. Also, in each column, the bold substring between the first and the second Kleene star is growing a pattern, $011100$ starting on branch $6$ on the first column and $110001$ starting on branch $10$ on the second column.

Let's look at the Collatz iterations (until they reach $\boldsymbol{1110}$) of $y_1= 10000(100011)^3 100(01)01$ which is a member of branch 2 and $y_2 = (100001001011110110)10000100101(111000)^211$ which is a member of branch 10:

\begin{align*}
y_1=&\phantom{0}100001000111000111000111000101 & y_2=\,& 1000010010111101101000010010111100011\\
\phantom{0}&110001101010101010101010101000 && 1100011100011100011100011100011010101\\
10&01010100000000000000000000 & 10&010101010101010101010101010101000000
\\
\boldsymbol{111}&\boldsymbol{0}0000 &\boldsymbol{1110}&00000000000000000000000000000
\end{align*}

Because they are both in $\E \Pred{3}{14}$ both strings $y_1$ and $y_2$ reach $\boldsymbol{1110}$ by using three odd steps (three applications of the map $T_1$). We display even steps without breaking a line, the Collatz process is ignoring all ending $0$s until it finds a $1$. Ending $0s$ are generated by the ending pattern $...010101$. Inner patterns within $y_1$ and $y_2$ seem to synchronize in order to generate ending $...010101$ patterns in future iterations and finally produce $\boldsymbol{1110}$. We believe that each branch of $\reg{3}{14}$ features a different ``mechanism'' that the Collatz process has in order to synchronize inner patterns so that they produce $\boldsymbol{1110}$ after $3$ odd steps. We leave as future work to understand those mechanisms which we believe are tightly connected to how carry propagates within $y_1$, $y_2$ and their Collatz descendants. We believe that the construction developed in this paper will provide a robust exploratory tool in order to support that future research.
\end{example}
\fi

\if 1\mode
\newcommand{\para}[1]{{\vspace{1.5ex}\noindent\bf #1.}}
\newpage
\para{Acknowledgements}

Many thanks to Jose Capco, Damien Woods, Pierre-Étienne Meunier and Turlough Neary for their kind help, interest and feedback on this project. We also thank Jeffrey C. Lagarias for his surveys on the Collatz problem (\cite{survey1} and \cite{survey2}). We thank the OEIS\footnote{\url{http://oeis.org/}}, always of great help. Finally, sincere thanks to \if 1\mode the \fi anonymous reviewers. Their comments were very helpful, for instance making us realise the exponential gain of our construction compared to previous literature.
\fi

\bibliographystyle{splncs04}
\bibliography{main}
\appendix

\if 0\mode
\clearpage
\section{Working with $\ZnZ{3^k}$}\label{app:structT}

In this Appendix, we recall the definition and main properties of groups $\ZnZ{3^k}$ and $(\ZnZ{3^k})^*$. Then, we proof two results on $T_{0,k}$ and $T_{1,k}$ (Definition~\ref{def:T}) which are used in the proof of Theorem~\ref{th:end}.

For $k>0$, we identify $\ZnZ{3^k}$ to $\{0,\ldots,3^k-1\} \subset \N$. Thus we implicitly have $x\in\ZnZ{3^k} \Rightarrow x < 3^k$. By $(\ZnZ{3^k})^{*}$ we refer to the multiplicative group of $\ZnZ{3^k}$, i.e. $(\ZnZ{3^k})^{*} = \{ x \in \ZnZ{3^k} \; | \; \exists y \;, \mud{xy}{1}{3^k} \}$. From elementary group theory results we can deduce that $(\ZnZ{3^k})^{*} = \{ x \; | \; x < 3^k \text{ and } x \text{ is not a multiple of } 3\}$. The element $2$ is thus always invertible in $\ZnZ{3^k}$. By $2^{-1}_{k}$ we refer to the modular inverse of $2$ in $\ZnZ{3^k}$, this means that in $\ZnZ{3^k}$ we have $2*2^{-1}_k = 1$. Furthermore, it is known that $2^{-1}_k$ is a primitive root of $(\ZnZ{3^k})^{*}$. This means that for all $x\in(\ZnZ{3^k})^{*}$ there exists $n\in\N$ such that $\mud{x}{(2^{-1}_k)^n = 2^{-n}_k}{3^k}$. Finally, even though $(\ZnZ{3^0})^{*} = \emptyset$ it will be useful, for the induction step of our main result (Theorem~\ref{th:main}), to take the convention $(2^{-1}_0)^0 = 2^0_0 = 0$.

We define $T_{0,k}$ and $T_{1,k}$ the analogous versions of $T_0$ and $T_1$ in $\ZnZ{3^k}$.

\defT

\begin{lemma}[Expression of $2^{-1}_k$]\label{lem:inv2}
We have: $ 2^{-1}_k = \frac{3^k+1}{2} $.
\end{lemma}
\begin{proof}

Let $z_k = \frac{3^k+1}{2} \in \N$. We have $z_k < 3^k$ thus $z_k\in\ZnZ{3^k}$. We also have, $2z_k = \mud{3^k + 1}{1}{3^k} $. Hence $z_k$ meets all the requirements to be $2^{-1}_k$.
\end{proof}

\begin{lemma}[Structure of $T_{0,k}$]\label{lem:t0}

Let $k\in\N$. For $x\in\ZnZ{3^k}$ we have:
\begin{equation*}
    T_{0,k}(x) = \begin{cases}
               		x/2 = T_{0}(x)  & \text{ if } x \text{ is even }\\
               		(3^k + x)/2  & \text{ if } x \text{ is odd }\\
           		\end{cases}
\end{equation*}
\end{lemma}

\begin{proof}

Let $x\in\ZnZ{3^k}$, we have $x<3^k$. Two cases:
\begin{itemize} 
	\item Case $x$ even. We have $x/2 < 3^k$ and $2*(x/2) = x$. This shows that $x/2 = 2^{-1}_k x$ and thus $T_{0,k}(x) = x/2$ for $x$ even.
	\item Case $x$ odd. We have $\mud{2^{-1}_k x}{\frac{3^k+1}{2}x}{3^k} $ by Lemma~\ref{lem:inv2}. Because $x=2y+1$, we have $2^{-1}_k x \equiv \frac{3^k+1}{2}(2y + 1) \equiv y + \frac{3^k+1}{2} \equiv \frac{3^k + 2y + 1}{2} \equiv (3^k + x)/2 \text{ mod } 3^k$. We also have $(3^k + x)/2 < 3^k$ thus, in $\ZnZ{3^k}$, $2^{-1}_k x = \frac{3^k + x}{2}$ and $T_{0,k}(x)=(3^k + x)/2$.
\end{itemize}

\end{proof}

\begin{lemma}[Structure of $T_{1,k}$]\label{lem:t1}

Let $k\in\N$. The function $T_{1,k+1}$ is $3^k$-periodic. Hence we simply have to characterize the behavior of $T_{1,k+1}$ on $\ZnZ{3^k}$. For $x\in\ZnZ{3^k}$ we have:
\begin{equation*}
    T_{1,k+1}(x) = \begin{cases}
               		(3^k + 3x + 1)/2 & \text{ if } x \text{ is even }\\
               		(3x + 1)/2 = T_{1}(x) & \text{ if } x \text{ is odd }\\
           		\end{cases}
\end{equation*}
\end{lemma}

\begin{proof}

For $x\in\ZnZ{3^k}$ we have: $T_{1,k+1}(x + 3^k) = \mud{2^{-1}_{k+1}(3^{k+1} + 3x + 1)}{2^{-1}_{k+1}(3x + 1)}{3^{k+1}} $. Thus the function $T_{1,k+1}$ is $3^k$-periodic. Now, two cases:

\begin{itemize} 
	\item Case $x$ odd. We have $(3x+1)/2 < 3^{k+1}$ and $2*((3x+1)/2) = 3x+1$. This shows that $(3x + 1)/2 = 2^{-1}_{k+1}(3x+1)$ and thus $T_{1,k+1}(x) = (3x + 1)/2$ for $x$ odd.
	\item Case $x$ even. We have $x=2y$ and $\mud{2^{-1}_{k+1}(3*2y+1)}{3y + 2^{-1}_{k+1}}{3^k}$. We have $3y + 2^{-1}_{k+1} = 3y + \frac{3^k + 1}{2} = (3^k + 3x + 1)/2$ by Lemma~\ref{lem:inv2}. Furthermore, $(3^k + 3x + 1)/2 < 3^{k+1}$ so we can conclude that $T_{1,k+1}(x) = (3^k + 3x + 1)/2$ when $x$ is even.

\end{itemize}

\end{proof}

\section{Feasible Vectors}\label{app:feasvec}

In this Section we present the formalism used in \cite{wirsching1998the} in order to prove Theorem~\ref{th:struct_occ}.

These results are based on a compact representation of parity vectors called \textit{feasible vectors} in \cite{wirsching1998the}:

\begin{definition}[Feasible vectors]\label{def:feasible}
The set of feasible vectors is $\F=\cup_{k=0}^{\infty}\N^{k+1}$. For a feasible vector $s=(s_0, \ldots, s_k)\in\F$, the \textit{length} of $s$, written $l(s)$ is $k$. The \textit{norm} of $s$ is $\norm{s} = l(s) + \sum_{i=0}^{l(s)} s_i$.
\end{definition}

\begin{example}
A feasible vector is a compact way to represent a parity vector. For instance, the parity vector $p=\;\downarrow\leftarrow\downarrow\downarrow\downarrow\leftarrow\leftarrow\downarrow\downarrow\; = \, (\downarrow)^1 \leftarrow (\downarrow)^3 \leftarrow (\downarrow)^0 \leftarrow (\downarrow)^2$. Will be represented by the feasible vector $s=(1,3,0,2)$. We have $\norm{p}=\norm{s}$ and $l(p)=l(s)$.
\end{example}

\begin{definition}[Backtracing Function]

Let $s=(s_0,\ldots,s_k)\in\F$, the backtracing function of $s$ is $v_s:\N \to \Q$ defined by:
$$ v_s(x) = T_0^{-s_0} \circ T_1^{-1} \circ T_0^{-s_1} \circ \ldots T_1^{-1} \circ T_0^{-s_k} $$

If $v_s(x)\in\N$ then we say that $s$ is backward feasible for $x$.

\end{definition}

\begin{lemma}[Lemma 2.17 in \cite{wirsching1998the}]

Let $s\in\F$ and $x\in\N$ such that $s$ is backward feasible for $x$. Then we have: $ T^{\norm{s}}(v_s(x)) = x $.

\end{lemma}

\begin{example}
For $p_3= (\downarrow)^3 \leftarrow (\downarrow)^0$, we have the corresponding feasible vector $s=(3,0)$.
\end{example}

Being a composition of affine functions, $v_s$ is affine. The author of \cite{wirsching1998the} completely characterises the structure of $v_s$:

\begin{lemma}[Lemma 2.13 in \cite{wirsching1998the}]\label{lem:heyho}

For $s=(s_0,\ldots,s_k)\in\F$ define:
\begin{align*}
c(s) = \frac{2^{\norm{s}}}{3^{l(s)}} &\text{ and }
r(s) = \sum_{j=0}^{k-1} \frac{2^{j+s_0+\dots+s_j}}{3^{j+1}}
\end{align*}

Then for any $x\in\N$ we have: $ v_s(x) = c(s)x - r(s)$.

\end{lemma}

Finally the following lemma of \cite{wirsching1998the} will essentially give the proof of Theorem~\ref{th:struct_occ}:

\begin{lemma}[Lemma 3.1 in \cite{wirsching1998the}]\label{app:lemocc}

Let $s\in\F$. Then there is exactly one $a < 3^{l(s)}$ such that for any $b\in\N$:
$$ s \text{ is backward feasible for } b\Leftrightarrow \mud{b}{a}{3^{l(s)}} $$

\end{lemma}

\begin{proof}

We know that: $ s \text{ is backward feasible for } b\Leftrightarrow v_s(b)\in\N $. Lemma~\ref{lem:heyho} gives: $v_s(b) = c(s)b - r(s) = \frac{1}{3^{l(s)}}\left ( 2^{\norm{s}}b - 3^{l(s)}r(s) \right )$. Hence, with $d=3^{l(s)}r(s)\in\N$: 
$$ s \text{ is backward feasible for } b\Leftrightarrow \mud{b}{2^{-\norm{s}}d}{3^{l(s)}}  $$

Because $2^{\norm{s}}$ is inversible in $\ZnZ{3^{l(s)}}$.

\end{proof}
Finally we can prove Theorem~\ref{th:struct_occ}:

\thstructocc*
\begin{proof}
Let $p\in\Pa$ and $s$ his associated feasible vector. By Lemma~\ref{app:lemocc}, we deduce that Point~\ref{point:so1} holds with the existence of $\alpha_{0,-1} < 3^{l(p)}$. From the same Lemma, we get $\alpha_{i,0} = 2^{\norm{p}}i + \alpha_{0,0}$. From Lemma~\ref{lem:heyho} we get: $\alpha_{i,0} = v_s(\alpha_{i,-1}) = v_s(3^{l(p)}i + \alpha_{0,-1}) = 2^{\norm{s}}i + v_s(\alpha_{0,0}) = 2^{\norm{s}}i + \alpha_{0,0}$. Finally, the bound $\alpha_{0,0} < 2^{\norm{p}}$ can also be derived from Lemma~\ref{lem:heyho}: $\alpha_{0,0} = v_s(\alpha_{0,-1}) = \frac{2^{\norm{s}}}{3^{l(s)}} \alpha_{0,-1} - r(s) < 2^{\norm{s}}$ since $r(s) \geq 0$ and $2^{\norm{s}} = 2^{\norm{p}}$.
\end{proof}

\section{Generating $\reg{k}{x}$ with \texttt{coreli}}\label{app:code}

We have implemented the construction of Theorem~\ref{th:main} in a Python library named \texttt{coreli}\footnote{Any similarity to Archangelo Corelli is purely coincidental, \url{https://www.youtube.com/watch?v=5BPhkY6xIP8}}, the Collatz Research Library: \url{https://github.com/tcosmo/coreli}. The long term goal of this library is to provide tools for exploring the Collatz process and also to implement constructions of past, current and future research on the Collatz problem. The library is fully documented here: \url{https://dna.hamilton.ie/tsterin/coreli/docs/}.

The construction of $\reg{k}{x}$ is performed in the following script: \url{https://github.com/tcosmo/coreli/blob/master/coreli/predecessors.py}. 

You'll find a detailed example (following Example~\ref{ex:final}) that will run you through the code's features in this notebook: \url{https://github.com/tcosmo/coreli/blob/master/examples/Binary\%20expression\%20of\%20ancestors\%20in\%20the\%20Collatz\%20graph.ipynb}.

\section{$\reg{4}{1}$}\label{app:long}

The following is the regular expression which defines $\E\Pred{4}{1}$, i.e. it recognises the binary representation -- with potential leading \texttt{0}s -- of any number $y$ that uses $4$ times the operator $T_1$ and any number of times the operator $T_0$ in order to reach $1$ in the Collatz process. Although we only have $k=4$, the regular expression $\reg{4}{1}$ is big: it is a tree with $11664$ branches.

$$\reg{4}{1} =$$ 

\begin{verbatim}(000000110010100100010110000111111001101011011101001111)*(((0)(0)((000010010
111101101)*(((0)(0)((000111)*(((0)(0)((01)*(((0)(1)((0)*)))))|((000)(1)((10)
*((()(1)((0)*))))))))|((000)(0)((100011)*(((10)(0)((01)*(((0)(1)((0)*)))))|(
(1000)(1)((10)*((()(1)((0)*))))))))|((0000100101111)(0)((110001)*(((110)(0)(
(01)*(((0)(1)((0)*)))))|((11000)(1)((10)*((()(1)((0)*))))))))|((000010010)(1
)((111000)*(((1110)(0)((01)*(((0)(1)((0)*)))))|(()(1)((10)*((()(1)((0)*)))))
)))|((0000100)(1)((011100)*(((01110)(0)((01)*(((0)(1)((0)*)))))|((0)(1)((10)
*((()(1)((0)*))))))))|((000010010111101)(1)((001110)*((()(0)((01)*(((0)(1)((
0)*)))))|((00)(1)((10)*((()(1)((0)*)))))))))))|((000000110010100100010)(1)((
100001001011110110)*(((10)(0)((000111)*(((0)(0)((01)*(((0)(1)((0)*)))))|((00
0)(1)((10)*((()(1)((0)*))))))))|((1000)(0)((100011)*(((10)(0)((01)*(((0)(1)(
(0)*)))))|((1000)(1)((10)*((()(1)((0)*))))))))|((10000100101111)(0)((110001)
*(((110)(0)((01)*(((0)(1)((0)*)))))|((11000)(1)((10)*((()(1)((0)*))))))))|((
1000010010)(1)((111000)*(((1110)(0)((01)*(((0)(1)((0)*)))))|(()(1)((10)*((()
(1)((0)*))))))))|((10000100)(1)((011100)*(((01110)(0)((01)*(((0)(1)((0)*))))
)|((0)(1)((10)*((()(1)((0)*))))))))|((1000010010111101)(1)((001110)*((()(0)(
(01)*(((0)(1)((0)*)))))|((00)(1)((10)*((()(1)((0)*)))))))))))|((000000110010
1)(0)((010000100101111011)*(((010)(0)((000111)*(((0)(0)((01)*(((0)(1)((0)*))
)))|((000)(1)((10)*((()(1)((0)*))))))))|((01000)(0)((100011)*(((10)(0)((01)*
(((0)(1)((0)*)))))|((1000)(1)((10)*((()(1)((0)*))))))))|((010000100101111)(0
)((110001)*(((110)(0)((01)*(((0)(1)((0)*)))))|((11000)(1)((10)*((()(1)((0)*)
)))))))|((01000010010)(1)((111000)*(((1110)(0)((01)*(((0)(1)((0)*)))))|(()(1
)((10)*((()(1)((0)*))))))))|((010000100)(1)((011100)*(((01110)(0)((01)*(((0)
(1)((0)*)))))|((0)(1)((10)*((()(1)((0)*))))))))|((01000010010111101)(1)((001
110)*((()(0)((01)*(((0)(1)((0)*)))))|((00)(1)((10)*((()(1)((0)*)))))))))))|(
(000000110)(0)((101000010010111101)*(((1010)(0)((000111)*(((0)(0)((01)*(((0)
(1)((0)*)))))|((000)(1)((10)*((()(1)((0)*))))))))|((101000)(0)((100011)*(((1
0)(0)((01)*(((0)(1)((0)*)))))|((1000)(1)((10)*((()(1)((0)*))))))))|((1010000
100101111)(0)((110001)*(((110)(0)((01)*(((0)(1)((0)*)))))|((11000)(1)((10)*(
(()(1)((0)*))))))))|((101000010010)(1)((111000)*(((1110)(0)((01)*(((0)(1)((0
)*)))))|(()(1)((10)*((()(1)((0)*))))))))|((1010000100)(1)((011100)*(((01110)
(0)((01)*(((0)(1)((0)*)))))|((0)(1)((10)*((()(1)((0)*))))))))|(()(1)((001110
)*((()(0)((01)*(((0)(1)((0)*)))))|((00)(1)((10)*((()(1)((0)*)))))))))))|((00
00001100101001000101100001111110011010110)(1)((110100001001011110)*(((11010)
(0)((000111)*(((0)(0)((01)*(((0)(1)((0)*)))))|((000)(1)((10)*((()(1)((0)*)))
)))))|((1101000)(0)((100011)*(((10)(0)((01)*(((0)(1)((0)*)))))|((1000)(1)((1
0)*((()(1)((0)*))))))))|((11010000100101111)(0)((110001)*(((110)(0)((01)*(((
0)(1)((0)*)))))|((11000)(1)((10)*((()(1)((0)*))))))))|((1101000010010)(1)((1
11000)*(((1110)(0)((01)*(((0)(1)((0)*)))))|(()(1)((10)*((()(1)((0)*))))))))|
((11010000100)(1)((011100)*(((01110)(0)((01)*(((0)(1)((0)*)))))|((0)(1)((10)
*((()(1)((0)*))))))))|((1)(1)((001110)*((()(0)((01)*(((0)(1)((0)*)))))|((00)
(1)((10)*((()(1)((0)*)))))))))))|((000000110010100100010110000111111)(0)((01
1010000100101111)*(((011010)(0)((000111)*(((0)(0)((01)*(((0)(1)((0)*)))))|((
000)(1)((10)*((()(1)((0)*))))))))|((01101000)(0)((100011)*(((10)(0)((01)*(((
0)(1)((0)*)))))|((1000)(1)((10)*((()(1)((0)*))))))))|(()(0)((110001)*(((110)
(0)((01)*(((0)(1)((0)*)))))|((11000)(1)((10)*((()(1)((0)*))))))))|((01101000
010010)(1)((111000)*(((1110)(0)((01)*(((0)(1)((0)*)))))|(()(1)((10)*((()(1)(
(0)*))))))))|((011010000100)(1)((011100)*(((01110)(0)((01)*(((0)(1)((0)*))))
)|((0)(1)((10)*((()(1)((0)*))))))))|((01)(1)((001110)*((()(0)((01)*(((0)(1)(
(0)*)))))|((00)(1)((10)*((()(1)((0)*)))))))))))|((00000011001010010001011000
01111110011)(0)((101101000010010111)*(((1011010)(0)((000111)*(((0)(0)((01)*(
((0)(1)((0)*)))))|((000)(1)((10)*((()(1)((0)*))))))))|((101101000)(0)((10001
1)*(((10)(0)((01)*(((0)(1)((0)*)))))|((1000)(1)((10)*((()(1)((0)*))))))))|((
1)(0)((110001)*(((110)(0)((01)*(((0)(1)((0)*)))))|((11000)(1)((10)*((()(1)((
0)*))))))))|((101101000010010)(1)((111000)*(((1110)(0)((01)*(((0)(1)((0)*)))
))|(()(1)((10)*((()(1)((0)*))))))))|((1011010000100)(1)((011100)*(((01110)(0
)((01)*(((0)(1)((0)*)))))|((0)(1)((10)*((()(1)((0)*))))))))|((101)(1)((00111
0)*((()(0)((01)*(((0)(1)((0)*)))))|((00)(1)((10)*((()(1)((0)*)))))))))))|((0
00000110010100100010110000111111001101)(0)((110110100001001011)*(((11011010)
(0)((000111)*(((0)(0)((01)*(((0)(1)((0)*)))))|((000)(1)((10)*((()(1)((0)*)))
)))))|((1101101000)(0)((100011)*(((10)(0)((01)*(((0)(1)((0)*)))))|((1000)(1)
((10)*((()(1)((0)*))))))))|((11)(0)((110001)*(((110)(0)((01)*(((0)(1)((0)*))
)))|((11000)(1)((10)*((()(1)((0)*))))))))|((1101101000010010)(1)((111000)*((
(1110)(0)((01)*(((0)(1)((0)*)))))|(()(1)((10)*((()(1)((0)*))))))))|((1101101
0000100)(1)((011100)*(((01110)(0)((01)*(((0)(1)((0)*)))))|((0)(1)((10)*((()(
1)((0)*))))))))|((1101)(1)((001110)*((()(0)((01)*(((0)(1)((0)*)))))|((00)(1)
((10)*((()(1)((0)*)))))))))))|((00000011001010010001011000011111100110101101
11010)(0)((111011010000100101)*(((111011010)(0)((000111)*(((0)(0)((01)*(((0)
(1)((0)*)))))|((000)(1)((10)*((()(1)((0)*))))))))|((11101101000)(0)((100011)
*(((10)(0)((01)*(((0)(1)((0)*)))))|((1000)(1)((10)*((()(1)((0)*))))))))|((11
1)(0)((110001)*(((110)(0)((01)*(((0)(1)((0)*)))))|((11000)(1)((10)*((()(1)((
0)*))))))))|((11101101000010010)(1)((111000)*(((1110)(0)((01)*(((0)(1)((0)*)
))))|(()(1)((10)*((()(1)((0)*))))))))|((111011010000100)(1)((011100)*(((0111
0)(0)((01)*(((0)(1)((0)*)))))|((0)(1)((10)*((()(1)((0)*))))))))|((11101)(1)(
(001110)*((()(0)((01)*(((0)(1)((0)*)))))|((00)(1)((10)*((()(1)((0)*)))))))))
))|((000000110010100100010110000)(1)((111101101000010010)*(((1111011010)(0)(
(000111)*(((0)(0)((01)*(((0)(1)((0)*)))))|((000)(1)((10)*((()(1)((0)*)))))))
)|((111101101000)(0)((100011)*(((10)(0)((01)*(((0)(1)((0)*)))))|((1000)(1)((
10)*((()(1)((0)*))))))))|((1111)(0)((110001)*(((110)(0)((01)*(((0)(1)((0)*))
)))|((11000)(1)((10)*((()(1)((0)*))))))))|(()(1)((111000)*(((1110)(0)((01)*(
((0)(1)((0)*)))))|(()(1)((10)*((()(1)((0)*))))))))|((1111011010000100)(1)((0
11100)*(((01110)(0)((01)*(((0)(1)((0)*)))))|((0)(1)((10)*((()(1)((0)*)))))))
)|((111101)(1)((001110)*((()(0)((01)*(((0)(1)((0)*)))))|((00)(1)((10)*((()(1
)((0)*)))))))))))|((0000001100101001000101100)(0)((011110110100001001)*(((01
111011010)(0)((000111)*(((0)(0)((01)*(((0)(1)((0)*)))))|((000)(1)((10)*((()(
1)((0)*))))))))|((0111101101000)(0)((100011)*(((10)(0)((01)*(((0)(1)((0)*)))
))|((1000)(1)((10)*((()(1)((0)*))))))))|((01111)(0)((110001)*(((110)(0)((01)
*(((0)(1)((0)*)))))|((11000)(1)((10)*((()(1)((0)*))))))))|((0)(1)((111000)*(
((1110)(0)((01)*(((0)(1)((0)*)))))|(()(1)((10)*((()(1)((0)*))))))))|((011110
11010000100)(1)((011100)*(((01110)(0)((01)*(((0)(1)((0)*)))))|((0)(1)((10)*(
(()(1)((0)*))))))))|((0111101)(1)((001110)*((()(0)((01)*(((0)(1)((0)*)))))|(
(00)(1)((10)*((()(1)((0)*)))))))))))|((0000001100101001000101100001111110011
01011011101001)(1)((101111011010000100)*(((101111011010)(0)((000111)*(((0)(0
)((01)*(((0)(1)((0)*)))))|((000)(1)((10)*((()(1)((0)*))))))))|((101111011010
00)(0)((100011)*(((10)(0)((01)*(((0)(1)((0)*)))))|((1000)(1)((10)*((()(1)((0
)*))))))))|((101111)(0)((110001)*(((110)(0)((01)*(((0)(1)((0)*)))))|((11000)
(1)((10)*((()(1)((0)*))))))))|((10)(1)((111000)*(((1110)(0)((01)*(((0)(1)((0
)*)))))|(()(1)((10)*((()(1)((0)*))))))))|(()(1)((011100)*(((01110)(0)((01)*(
((0)(1)((0)*)))))|((0)(1)((10)*((()(1)((0)*))))))))|((10111101)(1)((001110)*
((()(0)((01)*(((0)(1)((0)*)))))|((00)(1)((10)*((()(1)((0)*)))))))))))|((0000
001100101001000)(1)((010111101101000010)*(((0101111011010)(0)((000111)*(((0)
(0)((01)*(((0)(1)((0)*)))))|((000)(1)((10)*((()(1)((0)*))))))))|((0101111011
01000)(0)((100011)*(((10)(0)((01)*(((0)(1)((0)*)))))|((1000)(1)((10)*((()(1)
((0)*))))))))|((0101111)(0)((110001)*(((110)(0)((01)*(((0)(1)((0)*)))))|((11
000)(1)((10)*((()(1)((0)*))))))))|((010)(1)((111000)*(((1110)(0)((01)*(((0)(
1)((0)*)))))|(()(1)((10)*((()(1)((0)*))))))))|((0)(1)((011100)*(((01110)(0)(
(01)*(((0)(1)((0)*)))))|((0)(1)((10)*((()(1)((0)*))))))))|((010111101)(1)((0
01110)*((()(0)((01)*(((0)(1)((0)*)))))|((00)(1)((10)*((()(1)((0)*)))))))))))
|((000)(0)((001011110110100001)*(((00101111011010)(0)((000111)*(((0)(0)((01)
*(((0)(1)((0)*)))))|((000)(1)((10)*((()(1)((0)*))))))))|((0010111101101000)(
0)((100011)*(((10)(0)((01)*(((0)(1)((0)*)))))|((1000)(1)((10)*((()(1)((0)*))
))))))|((00101111)(0)((110001)*(((110)(0)((01)*(((0)(1)((0)*)))))|((11000)(1
)((10)*((()(1)((0)*))))))))|((0010)(1)((111000)*(((1110)(0)((01)*(((0)(1)((0
)*)))))|(()(1)((10)*((()(1)((0)*))))))))|((00)(1)((011100)*(((01110)(0)((01)
*(((0)(1)((0)*)))))|((0)(1)((10)*((()(1)((0)*))))))))|((0010111101)(1)((0011
10)*((()(0)((01)*(((0)(1)((0)*)))))|((00)(1)((10)*((()(1)((0)*)))))))))))|((
0000001100101001000101100001111)(1)((100101111011010000)*(((100101111011010)
(0)((000111)*(((0)(0)((01)*(((0)(1)((0)*)))))|((000)(1)((10)*((()(1)((0)*)))
)))))|((10010111101101000)(0)((100011)*(((10)(0)((01)*(((0)(1)((0)*)))))|((1
000)(1)((10)*((()(1)((0)*))))))))|((100101111)(0)((110001)*(((110)(0)((01)*(
((0)(1)((0)*)))))|((11000)(1)((10)*((()(1)((0)*))))))))|((10010)(1)((111000)
*(((1110)(0)((01)*(((0)(1)((0)*)))))|(()(1)((10)*((()(1)((0)*))))))))|((100)
(1)((011100)*(((01110)(0)((01)*(((0)(1)((0)*)))))|((0)(1)((10)*((()(1)((0)*)
)))))))|((10010111101)(1)((001110)*((()(0)((01)*(((0)(1)((0)*)))))|((00)(1)(
(10)*((()(1)((0)*)))))))))))|((000000110010100100010110000111111001101011011
)(1)((010010111101101000)*(((0100101111011010)(0)((000111)*(((0)(0)((01)*(((
0)(1)((0)*)))))|((000)(1)((10)*((()(1)((0)*))))))))|(()(0)((100011)*(((10)(0
)((01)*(((0)(1)((0)*)))))|((1000)(1)((10)*((()(1)((0)*))))))))|((0100101111)
(0)((110001)*(((110)(0)((01)*(((0)(1)((0)*)))))|((11000)(1)((10)*((()(1)((0)
*))))))))|((010010)(1)((111000)*(((1110)(0)((01)*(((0)(1)((0)*)))))|(()(1)((
10)*((()(1)((0)*))))))))|((0100)(1)((011100)*(((01110)(0)((01)*(((0)(1)((0)*
)))))|((0)(1)((10)*((()(1)((0)*))))))))|((010010111101)(1)((001110)*((()(0)(
(01)*(((0)(1)((0)*)))))|((00)(1)((10)*((()(1)((0)*)))))))))))|((0000001)(1)(
(001001011110110100)*(((00100101111011010)(0)((000111)*(((0)(0)((01)*(((0)(1
)((0)*)))))|((000)(1)((10)*((()(1)((0)*))))))))|((0)(0)((100011)*(((10)(0)((
01)*(((0)(1)((0)*)))))|((1000)(1)((10)*((()(1)((0)*))))))))|((00100101111)(0
)((110001)*(((110)(0)((01)*(((0)(1)((0)*)))))|((11000)(1)((10)*((()(1)((0)*)
)))))))|((0010010)(1)((111000)*(((1110)(0)((01)*(((0)(1)((0)*)))))|(()(1)((1
0)*((()(1)((0)*))))))))|((00100)(1)((011100)*(((01110)(0)((01)*(((0)(1)((0)*
)))))|((0)(1)((10)*((()(1)((0)*))))))))|((0010010111101)(1)((001110)*((()(0)
((01)*(((0)(1)((0)*)))))|((00)(1)((10)*((()(1)((0)*)))))))))))|((00000011001
0100)(1)((000100101111011010)*((()(0)((000111)*(((0)(0)((01)*(((0)(1)((0)*))
)))|((000)(1)((10)*((()(1)((0)*))))))))|((00)(0)((100011)*(((10)(0)((01)*(((
0)(1)((0)*)))))|((1000)(1)((10)*((()(1)((0)*))))))))|((000100101111)(0)((110
001)*(((110)(0)((01)*(((0)(1)((0)*)))))|((11000)(1)((10)*((()(1)((0)*)))))))
)|((00010010)(1)((111000)*(((1110)(0)((01)*(((0)(1)((0)*)))))|(()(1)((10)*((
()(1)((0)*))))))))|((000100)(1)((011100)*(((01110)(0)((01)*(((0)(1)((0)*))))
)|((0)(1)((10)*((()(1)((0)*))))))))|((00010010111101)(1)((001110)*((()(0)((0
1)*(((0)(1)((0)*)))))|((00)(1)((10)*((()(1)((0)*))))))))))))
\end{verbatim}

\fi

\end{document}